%% file: main.tex
\newif\ifextended%
\extendedtrue%

\ifextended
\documentclass[sigconf,authorversion]{acmart}
\else
\documentclass[sigconf]{acmart}
\fi 
\copyrightyear{2021}
\acmYear{2021}
\setcopyright{acmcopyright}\acmConference[WSDM '21]{Proceedings of the Fourteenth ACM International Conference on Web Search and Data Mining}{March 8--12, 2021}{Virtual Event, Israel}
\acmBooktitle{Proceedings of the Fourteenth ACM International Conference on Web Search and Data Mining (WSDM '21), March 8--12, 2021, Virtual Event, Israel}
\acmPrice{15.00}
\acmDOI{10.1145/3437963.3441825}
\acmISBN{978-1-4503-8297-7/21/03}

\input{packages}
\input{notation}

\begin{document}
\fancyhead{} 
\title[\titlefirstpart\ \titlesecondpart]{%
  \texorpdfstring{\titlefirstpart\\\titlesecondpart}{%
\titlefirstpart\ \titlesecondpart}}
\ifextended%
\subtitle{Extended version}
\fi

\author{Shahrzad Haddadan}
\affiliation{%
  \department{Dept.~of Computer Science \& Data Science Initiative}
  \institution{Brown University}
  \streetaddress{115 Waterman St.}
  \city{Providence}
  \state{RI}
  \postcode{02912}
  \country{USA}
}
\email{shahrzad\_haddadan@brown.edu}

\author{Cristina Menghini}
\affiliation{%
  \department{DIAG}
  \institution{Sapienza University}
  \streetaddress{Via Ariosto 25}
  \city{Rome}
  \postcode{00185}
  \country{Italy}
}
\email{menghini@diag.uniroma1.it}

\author{Matteo Riondato}
\orcid{0000-0003-2523-4420} 
\affiliation{%
  \department{Dept.~of Computer Science}
  \institution{Amherst College}
  \streetaddress{AC \#2232 Amherst College}
  \city{Amherst}
  \state{MA}
  \postcode{01002}
  \country{USA}
}
\email{mriondato@amherst.edu}

\author{Eli Upfal}
\affiliation{%
  \department{Dept.~of Computer Science}
  \institution{Brown University}
  \streetaddress{115 Waterman St.}
  \city{Providence}
  \state{RI}
  \postcode{02912}
  \country{USA}
}
\email{eli@cs.brown.edu}

\begin{teaserfigure}
  \textit{Democracy begins in conversation} --- John Dewey (attr.)
  \Description{Not a figure, just a quote.}
\end{teaserfigure}

\input{abstract}



\maketitle

\input{intro}
\input{related}
\input{prelims}
\input{algo}
\input{exper}
\input{concl}

\begin{acks}
Shahrzad Haddadan was supported by NSF Award CCF-1740741.
Part of Cristina Menghini's work was done while visiting Brown University and is supported by the ERC Advanced Grant 788893 AMDROMA.
Matteo Riondato is supported in part by \grantsponsor{NSF}{National Science
Foundation}{http://www.nsf.gov} award IIS-\grantnum{NSF}{2006765}. Eli Upfal was
supported in part by NSF awards RI-\grantnum{NSF}{1813444}, and
CCF-\grantnum{NSF}{1740741}. We thank an anonymous reviewer for correcting one
of our lemmas.
\end{acks}
\bibliographystyle{ACM-Reference-Format}
\bibliography{fairrandomwalks}
\ifextended%
\appendix%
\input{appendix}
\fi
\end{document}

%% file: packages.tex
\usepackage[utf8]{inputenc}
\usepackage{algorithm}
\usepackage{algpseudocode}
\usepackage{booktabs}
\usepackage{caption}
\usepackage{thmtools} 
\usepackage{cleveref}
\crefname{algorithm}{Alg.}{Algs.}
\Crefname{algorithm}{Algorithm}{Algorithms}
\crefname{appendix}{App.}{App.}
\Crefname{appendix}{Appendix}{Appendices}
\crefname{corollary}{Corol.}{Corolls.}
\Crefname{corollary}{Corollary}{Corollaries}
\crefname{conjecture}{Conjecture}{Conjectures}
\Crefname{conjecture}{Conjecture}{Conjectures}
\crefformat{conjecture}{Conjecture~#2#1#3}
\crefmultiformat{conjecture}{Conjectures~#2#1#3}{\ and~#2#1#3}{, #2#1#3}{\ and~#2#1#3}
\crefname{definition}{Def.}{Defs.}
\Crefname{definition}{Definition}{Definition}
\crefname{figure}{Fig.}{Figs.}
\Crefname{figure}{Figure}{Figures}
\crefname{lemma}{Lemma}{Lemmas}
\Crefname{lemma}{Lemma}{Lemmas}
\crefname{problem}{Prob.}{Probs.}
\Crefname{problem}{Problem}{Problems}
\crefname{proposition}{Prop.}{Props.}
\Crefname{proposition}{Proposition}{Propositions}
\Crefname{section}{Section}{Sections}
\crefname{section}{Sect.}{Sect.}
\crefname{subsection}{Sect.}{Sect.}
\Crefname{subsection}{Section}{Sections}
\crefname{subsubsection}{Sect.}{Sect.}
\Crefname{subsubsection}{Section}{Sections}
\crefname{table}{Table}{Tables}
\Crefname{table}{Table}{Tables}
\crefname{theorem}{Thm.}{Thms.}
\Crefname{theorem}{Theorem}{Theorems}
\usepackage{mathtools}
\usepackage{multirow}
\usepackage{nicefrac}
\usepackage{subcaption}

\theoremstyle{acmdefinition}
\newtheorem{problem}{Problem}

%% file: notation.tex
\DeclarePairedDelimiter\abs{\lvert}{\rvert}
\newcommand\algonameplain{RePBubLik}
\newcommand\algoname{\textsc{\algonameplain}}
\newcommand\algonameplus{\textsc{\algonameplain}+}
\DeclareMathOperator\argmax{\mathrm{arg}\max}
\newcommand\badletter{\mathcal{P}}
\newcommand\bad[1]{\badletter(#1)} 
\newcommand\badc[2]{\badletter_{#1}(#2)} 
\newcommand\badthres{r}
\newcommand\bias[1]{\rho(#1)}
\newcommand{\bubblechange}[5]{\Delta(#1,#2,#3,#4,#5)}
\newcommand\bubble[3]{\mathsf{B}^{#3}_{#1}\left(#2 \right)}

\newcommand\card[1]{\abs{#1}}
\newcommand{\centr}[3]{\mathsf{r}^{#1}\left(#2; #3 \right)}
\newcommand{\centrdict}{\mathcal{R}}
\newcommand{\cP}{\mathbb{P}}
\newcommand\expect[2]{\mathbb{E}_{#1} \left[#2 \right]}
\newcommand\Gnew{G_{\mathrm{new}}}
\newcommand\good[1]{\mathcal{Z}(#1)} 
\newcommand\goodthres{b}

\newcommand\optnode{\mathrm{opt}}
\newcommand\pol[2]{\mathcal{P}_{#1} \left(#2 \right)}
\newcommand\titlefirstpart{\algoname: Reducing the Polarized Bubble Radius}
\newcommand\titlesecondpart{with Link Insertions}
\ifcsname%
overunderset\endcsname%
\else%
\makeatletter
\newcommand{\overunderset}[3]{\binrel@{#3}%
\binrel@@{\mathop{\kern\z@#3}\limits^{#1}_{#2}}}
\makeatother
\fi%
\newcommand\towithin[4]{#1 \overunderset{#3}{#4}{\rightsquigarrow} #2}
\newcommand\nottowithin[4]{#1 \overunderset{#3}{#4}{\not\rightsquigarrow} #2}
\newcommand\trans{M}
\newcommand\ttime[3]{T^{#1}_{#2} \left(#3 \right)}

\newcommand{\vcolor}[1]{\mathsf{c}(#1)} 
\newcommand{\ovcolor}[1]{\bar{\mathsf{c}}(#1)} 
\newcommand\visitattime[2]{\Psi_{#1}(#2)} 

%% file: abstract.tex
\begin{abstract}

 The topology of the hyperlink graph among pages expressing different
 opinions may influence the exposure of readers to diverse content. Structural
 bias may trap a reader in a ``polarized'' bubble with no access to
 other opinions. We model readers' behavior as random walks. A node is in a
 ``polarized'' bubble if the expected length of a random walk from it to a page
 of different opinion is large. The structural bias of a graph is the sum of the
 radii of highly-polarized bubbles. We study the problem of decreasing the
 structural bias through edge insertions. ``Healing'' all nodes with high
 polarized bubble radius is hard to approximate within a logarithmic factor, so we
 focus on finding the best $k$ edges to insert to maximally reduce the
 structural bias. We present \algoname, an algorithm that leverages a variant of
 the random walk closeness centrality to select the edges to insert.  \algoname\
 obtains, under mild conditions, a constant-factor approximation. It reduces the
 structural bias faster than existing edge-recommendation methods, including
 some designed to reduce the polarization of a graph.

\end{abstract}

%% file: intro.tex
\section{Introduction}\label{sec:intro}
The World Wide Web  often contains thousands or even  millions of pages on every
topic, covering the whole spectrum of opinions. Exposure to \emph{diverse
content} is necessary to obtain a complete picture about a topic.
This exposure depends on the hyperlinks connecting the pages to each other. It
can be argued that enabling easier access to diverse content improves society as
it creates a more informed and less polarized general
public~\citep{benhabib1996toward}. Indeed politicians have strongly promoted and
even requested that audiences are exposed to varied content~\citep{LeFebvre17}.

The fact that diverse information is easily \emph{available} does not imply that
\emph{exploring} such diverse information is easy. Rather, echo chambers and
polarization on (social) media and
blogs~\citep{adamic2005political,conover2011political,flaxman2016filter} keep
the user in a \emph{homogeneous bubble}, exposing them only to agreeable
information~\citep{bakshy2015exposure}, and leading to conflicts between users
in different bubbles~\citep{KumarHLJ18,CossardDFMKMPS20}.

A web user can freely click on any hyperlink on the page they are currently
visiting, but the choice of which hyperlinks to include in the page is with the
website owner or editor, who, if not careful, may stop the user from being
exposed to diverse opinions. In other words, the hyperlink topology of a website
may suffer from \emph{structural bias} that traps the user in a bubble of
one-sided content without them knowing~\citep{horta2020youtube,menghini2020wikipedias}.
For example, structural bias on topic-induced networks, such as Wikipedia
topic-induced subgraphs, prevents users from building a well-rounded knowledge
about the topic. On query-/user-induced recommendation networks such as those on
Amazon and YouTube, structural bias hinders the discovery of diversified
content, reducing
serendipity~\citep{mouzhi2010serendipity,kotkov2016chalserendipity,anagnostopoulos2019principal}.
Structural bias thus limits the user's freedom while navigating the Web.

Socially-minded website editors would try to \emph{minimize  such structural
bias} by adding appropriate links to pages in the website. As an editor can only
modify a few pages and add only a few links to each edited page, they need to
carefully choose what page $p$ to edit, and what pages to link to from $p$.  The
goal of our work is to develop an algorithm that can give editors
recommendations for what links to add in order to reduce the structural bias.
There are key technical challenges that must be solved in order to give
effective link recommendations: \textit{1.} quantify the structural bias of a
page $p$, i.e., how hard it is to reach pages of a different opinion from $p$;
and \textit{2.} decide which of these pages should be linked from $p$. Existing
approaches to link recommendation, such as those based on vertex similarity,
fall short at this task because they are oblivious to the network structural
bias. Their recommendations often \emph{increase} the bias, rather than
decreasing it~\citep{MuscoMT18}, especially for highly controversial topics such
as political blogs (see also our experimental results in \cref{sec:exper}).
Thus, there is a need for a radically different approach to suggest links that
decrease the structural bias.

\paragraph{Contributions.} We study the problem of reducing the structural bias
of a graph by adding edges, and propose an algorithm, \algoname, to suggest such
edges. Our contributions are the following.

\begin{itemize}
  \item We consider directed graphs with vertices of two colors, representing a
    network of webpages on the same topic, with the two colors identifying the
    two opposite opinions on the topic, and edges representing links
    between pages. We define the \emph{(Polarized) Bubble Radius} (BR) of a
    vertex $p$ as a novel measure to quantify the structural bias of $p$ (see
    \cref{def:bubbleradius}), based on a task-specific variant of the hitting
    time for random walks, which models the navigation of a user on the web
    \citep{rwwithbackbutton,rwwithrestart}. The BR is the expected number of
    steps to go from $p$ to a page of different opinion, and can be easily
    estimated with a sampling-based approach with probabilistic guarantees
    (\cref{lem:bubbleapprox}), which enables us to tackle the first of the key
    challenges. 
  \item We define the \emph{structural bias} of a graph $G$ as the sum of the
    BRs of vertices with high BR (Eq.~\ref{eq:bias}). Completely removing the
    bias is APX-hard by reduction from set cover (see \cref{lem:reduction}).
    We therefore state the \emph{$k$-edge structural bias decrease} problem as
    the task of finding the set of $k$ pairs of vertices of different color
    such that adding the edge between the vertices in each pair would
    \emph{maximally} decrease the structural bias, over all possible sets of $k$
    pairs (see \cref{prob:kinsertions,thm:main}). This problem
    connects two areas: link recommendation and polarization reduction.
  \item We present \algoname, an efficient approximation algorithm for
    the $k$-edge structural bias decrease problem, that recommends the addition
    of $k$ edges between vertices of different color. Under mild conditions, the
    resulting decrease of the structural bias is within a constant factor of the
    optimal. Website editors have limited control on the probability that a
    newly added edge will be traversed by the users, so our algorithm makes no
    assumption or impose any restriction on it, as this probability is
    essentially external. At the core of \algoname\ is an analysis of the
    submodularity of the objective function (see~\cref{lem:submodular}),
    combined with the use of a task-specific variant of random-walk closeness
    \citep{rwclosenessWhiteSmyth}, a well-established centrality measure.
    \algoname\ requires good estimations of the random walk closeness, so we
    also give an approximation algorithm for this quantity (see
    \cref{lem:centr}).
  \item We evaluate \algoname\ on eight real datasets. We compare it to
    baselines and existing methods for edge recommendation either designed with
    the goal of reducing the controversy of a
    graphs~\citep{garimella2017reducing} or with the more general purpose of
    completing the network's link structure~\citep{grover2016node2vec}. Our
    algorithm leads to a faster reduction of the average BR (i.e., requiring
    fewer edge insertions) than existing contributions.
\end{itemize}

\ifextended%
The proofs of our theoretical results are in \cref{sec:appendix}.
\else
Due to space limitations, many of our proofs are in  the appendix of the extended online version~\citemissing.
\fi

%% file: related.tex
\section{Related work}\label{sec:related}
Polarization has long been studied in political
science~\citep{Sunstein02,Isenberg86}, and the recent diffusion of (micro-) blog
and social media platforms brought the issue to the attention of the broad
computer science community. Many works focused on showing the existence of
polarization on these
platforms~\citep{MoralesBLB15,adamic2005political,CossardDFMKMPS20,conover2011political,flaxman2016filter},
and on modeling, quantifying, and reducing
polarization~\citep{garimella2018quantifying,garimella2017reducing,MuscoMT18,ChitraM20,MakatosTT17,BeckerCDAG20,GarimellaGPT17,MatakosTG20,AslayMGG18,Akoglu14,GarimellaDFMGM18,MakatosTT17,NelimarkkaLS18,LiaoF14a,LiaoF14b,MunsonLR13},
or the glass  ceiling effect~\citep{StoicaC19,StoicaRC18,StoicaHC20}.
The literature is rich, to
the point that times seems ripe for an in-depth survey on the topic. Due to
space limitations, we discuss here only the relationship between our
work and the most relevant algorithmic contributions to polarization
reduction~\citep{garimella2017reducing,ChitraM20,MuscoMT18,AslayMGG18,BeckerCDAG20,GarimellaGPT17,MatakosTG20,menghini2019wikipol,menghini2020wikipedias,StoicaRC18}.

A first important difference of our work with respect to most previous
contributions is that they consider a network of \emph{users}, with edges
representing notions such as friendship or endorsement (e.g.,
retweets)~\citep{garimella2017reducing,ChitraM20,MuscoMT18,AslayMGG18,BeckerCDAG20,GarimellaGPT17,MatakosTG20,StoicaRC18}.
We focus instead on networks of \emph{content}, such as web pages linked to each
other, or products that are connected when similar. This deep difference makes
our contribution quite orthogonal to the ones in these previous works: we focus
on the polarization that is introduced by the topology of the network, rather
than on the polarizing effect of content on users or on the effect of users on
each other. We believe both aspects are important, but the structural bias we
focus on has only been subject to few
studies~\citep{menghini2019wikipol,menghini2020wikipedias}. These works, relying
on the notion of weighted reciprocity, propose a static and dynamic analysis of
structural bias on Wikipedia. The measure of structural bias we use is not
tailored to a specific website.

A second relevant difference from many previous works is that we consider the
``opinion'' of a page (i.e., a vertex) to be fixed, as it depends on its
content, while many previous contributions consider different models of user
opinion dynamics~\citep{MosselT17,DasGM14} to study the evolution of such
opinions as the users are exposed to different content or recommended different
friendships. The problem of recommending changes to the content of a page to
modify the opinion expressed in it is interesting but outside the scope of our
work. Instead, we focus on recommending the addition of links between pages, to
reduce the structural bias.

An interesting line of work studies how to reduce polarization in the content
seen by the users, by adapting information diffusion approaches through better
selection of the seed set for
cascades~\citep{AslayMGG18,BeckerCDAG20,GarimellaGPT17,MatakosTG20,StoicaHC20},
or by directly acting on recommendation systems~\citep{RastegarpahahGC2019}.
These methods can not be adapted to the problem we study, as they do not act on
the graph of content, but on that of users.

The most similar methods to ours are those that act on the structure of the
graph~\citep{ChitraM20,MuscoMT18,garimella2017reducing,StoicaRC18}, although as
we mentioned, they consider a network of users, not of content.
\citet{MuscoMT18} propose a network-design approach: they aim to find the best
set of edges between vertices such that the resulting graph would minimize both
disagreement and polarization. Rather than a ``design-from-scratch'' approach,
which seems mostly of theoretical relevance, we consider instead a practical
incremental approach that suggests modifications to an existing network.
Like us, \citet{garimella2017reducing} consider a graph polarization measure
based on random walks~\citep{garimella2018quantifying}. This measure essentially
quantifies the probability that a user of one opinion is exposed to content from
a user of a different opinion, thanks to a chain of retweets (represented by the
random walks). The measure is based on a variant of personalized PageRank for
sets of users with different opinions. The task requires to recommend new edges,
i.e., retweets, to increase this probability. Our measure of structural bias is
instead defined on the basis of
the (Polarized) Bubble Radius (BR) (\cref{def:bubbleradius}), which is a
vertex-dependent measure that represents the expected number of steps, for a
user starting at the page represented by vertex $v$, to reach, with a random
walk, a vertex with color different from $v$, representing a page expressing a
different opinion. Our measure is appropriate for our task of suggesting new
edges to make it easier for user to reach pages of different opinions. In
\cref{sec:exper} we compare our approach to that of
\citet{garimella2017reducing}.

An important line of work in graph analysis and mining looked at manipulating
the topology to modify different interesting characteristic quantities of
the graph, such as shortest paths and related
measures~\citep{parotsidis2015selecting,papagelis2011suggesting,demaine2010minimizing,perumal2013minimizing},
various forms of
centrality~\citep{parotsidis2016centrality,bergamini2018improving,d2019coverage,wkas2020manipulability,medya2018group,mahmoody2016scalable,angriman2020group},
and more~\citep{arrigo2016edge,arrigo2016updating,chan2014make,tong2012gelling,zeng2012manipulating}.
Despite the fact that we consider a specific centrality to choose the source of
the added edges, these methods cannot be used to solve our task of interest.

Another body of work related to ours are those which estimate graph properties using random walks \cite{Berarw,ChiHad2018rw,peres2018rw,avegragerw}. The studied properties are not defined based on random walks, rather random walks are used as a tool to estimate them. Here  based on random walks, we define a new property for networks: \emph{the structural bias},  and we use random walks to estimate it.  

%% file: prelims.tex
\section{Preliminaries}\label{sec:prelims}
Let $G=(V,E)$ be a directed weighted graph with $\card{V} = n$ vertices, such
that no vertex $v \in V$ has only incoming edges and no outgoing edges. $V$ is
partitioned in two disjoint sets $R$ and $B$ (i.e., $R \cap B = \emptyset$ and
$R \cup B = V$), called ``red'' or ``blue'' vertices, respectively. We denote
the color of a vertex $v$ by $\vcolor{v}$ and its opposite color by
$\ovcolor{v}$. The sets of all \emph{other} vertices of the same color as $v$ is
denoted as $C_v$ and the sets of all vertices of color different than $v$ is
denoted as $\bar{C}_v$.

The edge weights are \emph{transition probabilities}, as follows. Let $\trans$
be a $n \times n$ right-stochastic \emph{transition matrix} associated to $G$,
i.e., a matrix such that each entry $m_{i,j}$ is a probability, with $m_{i,j}=0$
if $(i,j) \notin E$, and such that $\sum_{j=1}^n m_{i,j}=1$.

We are interested in random walks on the graph $G$ using the transition matrix
$\trans$. Intuitively, a random walk starting at a vertex $v$ explores the graph
by choosing at each step an outgoing edge from the current vertex, with
probability equal to the weight of such edge, independently from previous
choices. Let $S \subseteq V$ and $v \in V$. Let $T_v(S)$ be the random variable
indicating the first instant when a random walk from $v$ hits (i.e., reaches) any
vertex in $S$. The quantity $\expect{G}{T_v(S)}$ is known as the \emph{hitting
time of $S$ from $v$}, where the expectation is over the space of all random
walks on $G$ starting from $v$, with transition probabilities given by $\trans$.

Variants of random walks, such as random walks with restarts or with back
button, are widespread models for network
exploration~\citep{rwwithbackbutton,rwwithrestart}. It is realistic to assume
that there is an upper bound $t$, which we call the \emph{exploration
factor}, on the length of a walk performed by the users. For example, we can
assume that there is an upper limit on the number of pages that a user will
visit one after the other in a browsing session. The value of the parameter $t$
can be derived, for example, from traces of visits. In most practical cases, $t$
is likely to be bounded by a polylogarithmic quantity in the number of nodes, if
not a constant.

For a random walk starting from $v \in V$, given a set $S \subseteq V$, we
define the random variable $\ttime{t}{v}{S}$ as $\min \{t, T_v(S)\}$.
This variable is more appropriate for measuring the length of browsing sessions,
which have bounded length, than the unbounded length classically used when
discussing random walks.

For a graph $Z$, any vertex $u$, and any set $S$ of vertices, let
$\towithin{u}{S}{\mathsf{cond}}{Z}$, denote the event that a random walk in $Z$
from $u$ hits a vertex in $S$ without first visiting any vertex in $\bar{C}_u$
and while satisfying the condition $\mathsf{cond}$ on the number of
steps needed to hit  $S$. For example, $\towithin{u}{S}{< t}{Z}$ is the event
that a random walk in $Z$ from $u$ hits a vertex in $S$ in \emph{less than $t$
steps}, without first visiting any vertex in $\bar{C}_u$. We denote the
complementary event as $\nottowithin{u}{S}{\mathsf{cond}}{Z}$.

\subsection{Random-Walk Closeness Centrality}\label{sec:prelims:rwcc}
We adapt the definition of the standard random-walk closeness
centrality~\citep{rwclosenessWhiteSmyth} to bounded random walks so that the
contribution to the centrality of $v$ by vertices that do not reach $v$ in less
than $t'$ steps (in expectation) is zero, for any $t'$.

\textbf{Random-walk closeness centrality (bounded form).} For a vertex $v \in
V$, and any $t'$, the \emph{$t'$-bounded} Random Walk Closeness Centrality
(RWCC) measure with respect to subset $S \subseteq V$ is
\begin{align*}
  \centr{t'}{v}{S} &\doteq \frac{1}{\card{S}} \sum_{w \in S} \left( t' -
  \expect{G}{\ttime{t'}{w}{v}} \right) \\
  & = \frac{1}{\card{S}} \sum_{w \in S}
  \sum_{i=1}^{t'} (t'- i) \cP\left( \towithin{w}{v}{=i}{G} \right) \enspace.
\end{align*}
Computing the exact RWCC is expensive. To estimate $\centr{t'}{v}{S}$, we
pick $z$ vertices ${\{w_i\}}_{i=1}^r$ u.a.r.\ from $S$, and run some $\kappa$
random walks to obtain an estimate $\bar{h}_{w_{i}}$ of
$\expect{G}{T^{t'}_{w_i}(v)}$ for each $w_i$. The quantity  $\bar{r}(v) \doteq
t' - \nicefrac{1}{z} \sum_{i=1}^z \bar{h}_{w_{i}}$ is a good approximation of
$\centr{t'}{v}{S}$.

\begin{restatable}{lemma}{lemcentr}\label{lem:centr}
  Let $z \geq {(\nicefrac{t'}{2\varepsilon})}^2 \delta^{-1}$. Then
  \[
    \cP\left(\card{\bar{r}(v) - \centr{t'}{v}{S}} \geq \varepsilon \right) \leq
    \delta \enspace.
  \]
\end{restatable}

%% file: algo.tex
\section{Bubble Radius and Structural Bias}\label{sec:bubblebias}
We introduce the \emph{(Polarized) Bubble Radius} to quantify how \emph{likely}
users starting their random walk on a vertex $v \in V$ of one color, are to hit
a vertex of the other color in at most $t$ steps.

\begin{definition}\label{def:bubbleradius}
  The \emph{(Polarized) Bubble Radius (BR) $\bubble{G}{v}{t}$ of $v$ with
  exploration parameter $t$} is
  \[
    \bubble{G}{v}{t} \doteq \expect{G}{\ttime{t}{v}{\bar{C}_v}} \enspace.
  \]
\end{definition}

A random walk starting at a vertex $v$ with high BR is unlikely to hit a vertex
in $\bar{C}_v$ in  fewer-than-or-exactly $t$ steps. The following
lemma formalizes this idea on common models for web browsing (random walks with
restarts or with back button~\citep{rwwithbackbutton,rwwithrestart}).

\begin{restatable}{lemma}{lemrestart}\label{lem:restart}
  Let $r \in \mathbb{N}$, and consider a user who starts their random walk at
  $v \in V$ and may either restart their walk from $v$ or hit the back button up
  to $r$ times. Let $\mathcal{T}_v$ be the random variable denoting the number
  of steps such user takes to hit a vertex in $\bar{C}_v$. If
  $\bubble{G}{v}{t} \geq t (1 - \nicefrac{1}{8r})$, then
  $\cP\left(\mathcal{T}_v\leq t/2 \right) \leq 1/4$. If instead
  $\bubble{G}{v}{t} \leq b$ for some $b > 0$, then $\cP\left(\mathcal{T}_v > 4br
  \right) \leq 1/4$.
\end{restatable}

Given $t$, it is easy to estimate $\bubble{G}{v}{t}$ for each vertex $v
\in V$ by sampling random walks from $v$. The following result, whose proof uses
the Hoeffding's bound and the union bound, shows the trade-off between the
number of sampled random walks and the accuracy in estimating the BR of $v$.

\begin{restatable}{lemma}{lembubbleapprox}\label{lem:bubbleapprox}
  For each $v \in V$, let $w^{(v)}_1,w^{(v)}_2,\dotsc, w^{(v)}_r$ be $r$ random
  walks from $v$ and stopped either when they hit a vertex of color
  $\ovcolor{v}$ or when they run for $t$ steps, whichever happens first. For
  $i=1,\dotsc,r$, let $b^{(v)}_i$ be the length of random walk $w^{(v)}_i$. Let
  \[
    \bar{B}(v) \doteq \frac{1}{r} \sum_{i=1}^r b^{(v)}_i \enspace.
  \]
  Let $\varepsilon, \delta \in (0,1)$. If $ r \ge
  \frac{t^2}{\varepsilon^2} \ln \frac{2n}{\delta},
  $ then
  \[
    \cP \left( \exists v \in V\ \text{s.t.}\ \abs{\bubble{G}{v}{t} -
    \bar{B}(v)} > \varepsilon \right) < \delta,
  \]
  where the probability is over the choice of the random walks.
\end{restatable}

In the rest of the work, we assume for simplicity to have access to the
\emph{exact} BR of every vertex. The above result makes this assumption
reasonable because computing approximations of extremely high quality is
relatively inexpensive.

\paragraph{The structural bias}
On the basis of the BR, we define two sets of vertices: \emph{cosmopolitan} and
\emph{parochial}. Given two reals $\goodthres$ and $\badthres$ with
$1\le \goodthres < \badthres \le t$, the set $\good{G}$ of \emph{cosmopolitan}
vertices contains all and only the vertices in $G$ with BR \emph{at most}
$\goodthres$, and the set $\bad{G}$ of \emph{parochial} vertices contains all
and only the vertices in $G$ with BR at least $\badthres$. For ease of
notation, we do not include $\goodthres$ and $\badthres$ in the notation for
$\good{G}$ and $\bad{G}$. In the rest of this work, we assume for simplicity
$\goodthres=2$ and $\badthres=t/2$, but this assumption can be easily removed.
$\good{G}$ and $\bad{G}$ are \emph{disjoint}, but they do not necessarily form a
partitioning of $V$. We will often consider the partitioning of $\bad{G}$ by
color, i.e., the two sets $\badc{R}{G}$ and $\badc{B}{G}$, containing  the
parochial vertices of color $R$ or $B$ respectively.

\begin{definition}
  The \emph{structural bias} $\bias{G}$ of $G$ is the sum of the BRs of the
  parochial nodes of $G$, i.e.,
  \begin{equation}\label{eq:bias}
    \bias{G} \doteq \sum_{v \in \bad{G}} \bubble{G}{v}{t} \enspace.
  \end{equation}
\end{definition}
It is reasonable to consider only the parochial nodes in the definition of
structural bias because they are the ones such that a random walk from them is
very unlikely to hit any vertex of color different than the starting vertex
(see also \cref{lem:restart}).

Our goal in this work is to find a set of edges with extrema of different color
whose addition to $G$ would decrease the structural bias of the network. It is
reasonable to only consider edge with extrema of different color, as they are
always preferable (i.e., will result in a higher decrease of the structural
bias) than edges with monochromatic extrema: the addition of the new edge can
only have positive impact on the parochial vertices of the same color as the
source, and has no impact on the parochial vertices of the other color. If we
could add \emph{any number} of such edges to $G$, it would be easy to bring the
structural bias of $G$ to zero, as there would be no parochial nodes left. This
assumption is not realistic: the number of links that a website editor can add
to a single page and to the whole graph is limited by many factors, such as the
fact that a human-readable page cannot have too many links, and the fact that
the editor can only spend a limited time on this activity. Nevertheless, ideally
one would want to solve the following problem.

\begin{problem}\label{prob:zerobias}
  Given a color $C \in \{R, B \}$, find the  smallest set $A$ of pairs of
  distinct edges $(v, w) \notin E$ with $\vcolor{v} = C$ and $\vcolor{w} \neq C$
  such that, for the graph $\Gnew = (V, E \cup A) $ it holds
  $\badc{C}{\Gnew}=\emptyset$.
\end{problem}

\begin{restatable}{lemma}{lemreduction}\label{lem:reduction}
  \Cref{prob:zerobias} is NP-hard and APX-hard.
\end{restatable}

\section{Reducing the BR with insertions}\label{sec:algo}
Since \cref{prob:zerobias} is hard to even approximate (\cref{lem:reduction}),
we seek to answer a close relative (\cref{prob:kinsertions}). We first introduce
a set of measures to capture the change in the BRs of the (original) parochial
nodes of $G$ after edge insertions. Let $\Gnew$ be obtained from $G$ by
inserting a set $\Sigma$ of directed edges between nodes of different colors,
with each inserted edge $e=(v,w)$ having weight $m_e$ (also denoted as
$m_{vw}$). For a set $U$ of vertices, we define the \emph{gain} of $U$ due to
$\Sigma$ as
\[
  \bubblechange{G}{U}{{\Sigma}}{{\{m_{e}\}}_{e \in \Sigma }}{t'} \doteq\
  \frac{1}{\card{U}} \sum_{u \in U} \left( \bubble{G}{u}{t'} -
  \bubble{\Gnew}{u}{t'} \right) \enspace.
\]
When adding an edge to the graph, we also have to decide its weight. It seems
excessive to assume complete freedom in choosing the weight. We make the
assumption that the weight $m_{vw}$ of an edge $(v,w)$ that we would like to add
is given to us by an oracle which computes $m_{vw}$ only as a function of $v$
and of information \emph{local} to $v$ (e.g., its out-degree) obtained from $G$
and potentially a set of other edges (and their weights) that we want to add
from $v$. The weight $m_{vw}$ is the probability that a random walk arriving at
$v$ will move to $w$ in the next step. When adding $(v,w)$ with weight $m_{vw}$,
the other edges outgoing from $v$ have their weights multiplied by $1-m_{vw}$ to
ensure that the sum of the weights of the edges leaving $v$ is $1$.

The problem we want to solve then is the following.

\begin{problem}\label{prob:kinsertions}
  In graph $G$, let $C$ be either blue or red. Find a set
  ${\Sigma}={\{(v_i, w_i)\}}_{i=1}^k$ of $k$ edges whose source vertices all
  have color $C$ and all destination vertices have the other color, that
  maximizes $\bubblechange{G}{\pol{C}{G}}{\Sigma}{{\{m_{e}\}}_{e \in
  {\Sigma}}}{t}$.
\end{problem}

\algoname\ (\cref{alg:repbublik}) is our algorithm to approximate
\cref{prob:kinsertions}. Before describing it in detail, we give an intuition of
its workings, and present the theoretical results that guided its design.
Specifically, since our objective function is \emph{monotonic and submodular}
(\cref{lem:submodular}), we can greedily choose the edges to be added one by
one. Due to our oracle assumption on the weights, any vertex of color different
than the source can be picked as the target of the added edge, so the problem
essentially \emph{reduces to finding the sources for the edges to be added}.
\Cref{lem:gain} quantifies the gain when picking each source according to a
specific measure depending on the bounded RWCC and on the oracle-given weight
that only depends on the source. In \cref{lem:opt} we show that under mild
conditions this choice is constantly close to an optimal choice. The following
theorem states the approximation qualities of \algoname.

\begin{theorem}\label{thm:main}
  Let $\Sigma$ be the output of \algoname\ and $\mathsf{OPT}$ be the optimal
  solution to \cref{prob:kinsertions}. Let $\Delta_\Sigma =
  \bubblechange{G}{\pol{C}{G}}{\Sigma}{{\{m_e\}}_{e \in \Sigma}}{t}$
  Then
  \[
    \bubblechange{G}{\pol{C}{G}}{\mathsf{OPT}}{{\{m_e\}}_{e \in \mathsf{OPT}}}{t}
    \leq\left( 4 \gamma(G)+1\right) \left( 1 + \frac{1}{e} \right)
    \Delta_{\Sigma},
  \]
  where $\gamma(G)$ is the maximum over all $u \in V$ of sum of the
  probabilities, for $i=0,\dotsc,t-1$, that a random walk starting at $u$ visits
  $u$ at step $i$ without first visiting a vertex in $\bar{C}_u$ (see
  also~\eqref{eq:fprob}), 
  which is a constant for many graphs.
\end{theorem}

We now proceed towards presenting lemmas which together provide a proof for
\cref{thm:main}. \Cref{lem:bubbleapprox,lem:centr} provide bounds of order
$\Theta(nt^2)$ on the runtime of the pre-processing phases of \algoname.
Therefore for small values of $t$, \algoname\ is more efficient than algorithms
that compute hitting times using the Laplacian, which need $\Omega(n^3)$ steps.

For any vertex $v$, and $0 \le i \le t$, let $\visitattime{v}{i}$ be the
probability that a random walk (in $G$) from $v$ visits $v$ at step $i$ before
reaching a vertex in $\bar{C}_v$ (it holds $\visitattime{v}{0}=1$ and
$\visitattime{v}{1}=0$ for every $v$). For any $t' \le t$, let
\begin{equation}\label{eq:fprob}
  \mathcal{F}_{t'}(v) = \sum_{i=0}^{t'-1} \visitattime{v}{i}
  \enspace.
\end{equation}
The following lemma shows upper and lower bounds to the change in the bubble
radius of a vertex when an new edge from it is added to the graph.

\begin{restatable}{lemma}{lemgainbounds}\label{lem:gainbounds}
  Let $v \in \pol{}{G}$,  $w \in \bar{C}_v$ and $t'\leq t$. Let $\Gnew$ be
  the graph obtained after adding $e=(v,w)$ to $G$, with weight $m_{e}$.
  The gain $\bubblechange{G}{v}{e}{m_e}{t'}$ is such that
  \[
    \left( \bubble{G}{v}{t'}
    - 1 \right) m_e \leq    \bubblechange{G}{v}{e}{m_e}{t'}\\
    \leq \mathcal{F}_{t'}(v) \left( \bubble{G}{v}{t'}
    - 1 \right) m_e \enspace.
  \]
\end{restatable}

Decreasing the BR of $v$ decreases the BRs of vertices in $C_v$ close to $v$,
and thus the whole network. \Cref{lem:allvertices} quantifies this change.

\begin{restatable}{lemma}{lemallvertices}\label{lem:allvertices}
 Let $e = (v,w)$ be the edge with weight $m_e$ added to $G$ to obtain $\Gnew$.
 For any other vertex $u \in \pol{C_v}{G}$, it holds
\[
  \bubblechange{G}{u}{e}{m_e}{t} =
  \sum_{i=1}^{t-2} \left( \bubblechange{G}{v}{e}{m_e}{t - i}
  \cP\left( \towithin{u}{v}{=i}{G} \right)  \right)\enspace.
 \]
\end{restatable}

Recall that our greedy choice is to identify a node $v$ that maximizes the gain
$\bubblechange{G}{\pol{}{G}}{(v,w)}{m_{v}}{t}$ where $w$ is any vertex in
$\bar{C}_v$.  \Cref{lem:allvertices} suggests that a good candidate $v$ is a
vertex that is likely to be reached by short random walks from many other
vertices in $\pol{C_v}{G}$, a property that is captured by the bounded RWCC
$\centr{t-2}{v}{\pol{C_v}{G}}$ (\cref{sec:prelims:rwcc}).

%
%
%
%


Now, we first quantify the gain for adding an edge from any vertex with RWCC $c$
(\cref{lem:gain}). Then we show that under mild conditions on the return time of
vertices we get a constant approximation by greedily choosing a vertex with
maximum ${\rm RWCC }\times m_v$ (\cref{lem:opt}).

\begin{restatable}{lemma}{lemgain}\label{lem:gain}
  Let $v \in \pol{}{G}$. Let $w \in \bar{C}_v$, and assume to add the edge
  $e=(v,w)$ with weight $m_e$.
  It holds
  \[
    \bubblechange{G}{\pol{C_v}{G}}{e}{m_e}{t} \ge \frac{m_e}{2}
    \centr{t-2}{v}{\pol{C_v}{G}} \enspace.
  \]
\end{restatable}

This lemma suggests that inserting edges from a vertex $v$ with the highest
value of $m_v \centr{t-2}{v}{\pol{C}{G}}$ \emph{may} result in a larger
improvement in the objective function than if we chose a different source. In
the next lemma we compare the effect of choosing such sources to the effect of
an optimal choice.

\begin{restatable}{lemma}{lemopt}\label{lem:opt}
  Consider the set $\pol{C}{G}$ where $C$ is either color. Among all vertices in
  $\pol{C}{G}$ let $v$ and $\optnode$ be
  \begin{align*}
    \optnode &= \argmax_{u \in \pol{C}{G}}
    \bubblechange{G}{\pol{C}{G}}{e_u}{m_{u}}{t},\\
    v &= \argmax_{u \in \pol{C}{G}} m_{u} \centr{t-2}{u}{\pol{C}{G}},
  \end{align*}
  where   $e_u$  is any potentially inserted  edge connecting $u$ to
  $\bar{C}_u$, and $m_u$ is its weight.\footnote{Our assumption on the oracle
    giving the weight ensures that $m_u$ only depends on $u$, not on the target
  of $e_u$.}
  It holds
  \[
    \bubblechange{G}{\pol{C}{G}}{e_\optnode}{m_\optnode}{t}
    \leq (4 \gamma(G) + 1) \bubblechange{G}{\pol{}{G}}{e_v}{m_v}{t},
  \]
  where $\gamma(G) = \max_{u \in G} \mathcal{F}_t (u)$.
\end{restatable}

If the probability of getting back to $u$ in less than $t$ steps is less than
$\alpha$ for some constant $\alpha$ then $\gamma(G) \leq \alpha$. 
This assumption is realistic since $t$ is usually small and the return time to
$u$ is often much larger than $t$.

Finally, we show that the gain function 
is monotonic and sub-modular.

\begin{restatable}{lemma}{lemsubmodular}\label{lem:submodular}
  Let $C$ be either blue or red and $v, u \in \pol{C}{G}$, and $w_v, w_u \in
  \bar{C}_v$, such that $e_v = (v, w_v)$ and  $e_v = (u, w_u)$  are not existing
  edges. Let $\Sigma = \{e_v, e_u\}$. It holds
  \ifextended%
    \begin{equation}\label{eq:monotonicity}
  \else
  \[
  \fi
    \bubblechange{G}{\pol{C}{G}}{e_v}{m_{e_v}}{t} \leq
  \bubblechange{G}{\pol{C}{G}}{\Sigma}{{\{m_e\}}_{e \in \Sigma}}{t},
  \ifextended%
    \end{equation}
  \else
    \]
  \fi
  and
    \begin{align}
      \bubblechange{G}{\pol{C}{G}}{\Sigma}{{\{m_e\}}_{e \in \Sigma}}{t} \leq &
    \bubblechange{G}{\pol{C}{G}}{e_v}{m_{e_v}}{t} \nonumber \\
    & + \bubblechange{G}{\pol{C}{G}}{e_u}{m_{e_u}}{t} \enspace.%
    \ifextended\label{eq:submod}\else\nonumber\fi
  \end{align}
\end{restatable}

We are now ready to prove \cref{thm:main}.

\begin{proof}[Proof of \cref{thm:main}]
  \Cref{lem:submodular}  shows the monotonicity and submodularity of the
  objective function. Thus, a greedy algorithm that picks, iteratively, the $k$
  best choices over all parochial vertices of color $C$ as the sources of the
  added edges, will result in a $(1 + \nicefrac{1}{e})$-approximation.
  \Cref{lem:gain,lem:opt} show that by choosing a vertex $v$  maximizing $m_v
  \centr{t-2}{v}{\pol{C}{G}}$ among all parochial vertices of color $C$, we
  obtain a vertex such that the gain when adding an edge from this source is a
  $4 \gamma(G)+1$-approximation to the greedy choice. Thus, the correctness of our
  algorithm is concluded by putting these lemmas together.
\end{proof}

\begin{algorithm}[t]
\algrenewcommand\algorithmicindent{1.0em}
\begin{algorithmic}[1]

\State\textbf{Input}: Graph $G=(V,E)$, desired insertions $k_C$, oracle
$\mathcal{W}_G:V \times 2^{V \times V} \to [0,1]$, $C\in\{R,B\}$.
\State\textbf{Output}: Set $\Sigma_C$ of $k_C$ edges to be inserted, with their weights.
\State$\Sigma_C \gets \emptyset$
\For{$i=1$ \textbf{to} $k_C$}
    \State$P \gets$ \texttt{computeParochials}($G \cup \Sigma_C$, $C$) 
    \State$\centrdict \gets$ \texttt{computeRWCentrality}($P$, $G \cup \Sigma_C$) 
    \State$v_i \gets {\rm argmax}_{v \in P} \centrdict(v)\times
    \mathcal{W}_G(v, \Sigma_C) $
    \State$u_i\gets $ arbitrary  in $ \bar{C}_{v_i}$
    \State$\Sigma_C\gets \Sigma_C\cup\{(v_i,u_i)\}$
\EndFor%
\State\textbf{return} $\Sigma_C$ 
\end{algorithmic}
\caption{\algoname}\label{alg:repbublik}
\end{algorithm}

We can now give the details to \algoname. The algorithm takes as input the graph
$G$, the number $k_C$ of desired edge insertions, the oracle $\mathcal{W}$ that
determines the weights of the new edges, and the set of nodes $C$. It
first creates the empty set $\Sigma_C$ that will store the edges to be added and then enters
a for loop to be repeated for $k_C$ times. At every iteration of the loop, it first
computes the BR of every node in $C$ in the graph (denoted in the pseudocode as $G \cup \Sigma_C$) obtained by adding to $G$ the edges currently in $\Sigma_C$ (with their weights obtained from the oracle $\mathcal{W}_G$) (in practice, the BR is computed using the approximation
algorithm outlined in \cref{lem:bubbleapprox}). Thanks to this computation, the algorithm obtains (line 5) the set $P$ of parochial nodes in this graph (at the first iteration of the loop $P=\badc{C}{G}$). It then obtains the centralities values $\centr{t-2}{v}{P}$ of every node $v \in P$ (in practice, using the approximation algorithm outlined in
\cref{lem:centr}), storing them in a dictionary $\centrdict$ (line 6). 
%
%
The algorithm then selects the node $v_i \in P$ associated to the
maximum quantity $\centrdict(v_i)\times \mathcal{W}_G(v_i, \Sigma_C)$, and
arbitrarily picks a node $u_i$ of the opposite color of $v_i$ (i.e., of the color other than $C$). The directed edge $(v_i,u_i)$ is added to the set $\Sigma_C$ (lines 7--9). After $k_C$ iterations of the loop,
the algorithm returns $\Sigma_C$, together with the weights obtained from the oracle.

\algoname would require a re-computation of the BRs and of the centralities of all vertices, at every iteration of the loop, which would require to run a very large number of random walks, making it computationally very expensive. We now propose a more practical alternative \algonameplus, at the price of losing the approximation guarantees. \algonameplus\ only computes $\bad{C}{G}$ and $\centrdict$ before entering the for loop, and uses the same values throughout its execution, but trades off the consequences of this choice by adding a penalty factor to the objective function involved in the selection of the source vertices for the edges to be added. Specifically, \algonameplus\ chooses $v_i$ (line 7) by maximizing the quantity $\nicefrac{\centrdict(v) \times \mathcal W_G(v, \Sigma_C)}{\eta_v}$, where $\eta_v$ is a penalty factor equals to one plus the number of edges with source $v$ in $\Sigma_C$ (thus at iteration 1, $\eta_v = 1$ for every node). This penalty factor favours the insertion of edges from nodes that have not yet been altered. Consequently, it indirectly \textit{(1)} handles the possibility that nodes with new edges are no longer parochial, thus we want to avoid to keep adding edges to them; and \textit{(2)} avoids that the new edges are added from a restricted set of nodes, limiting the positive effect of the insertions on $\bubblechange{G}{\badc{C}{G}}{{\Sigma_C}}{{\{m_{e}\}}_{e \in \Sigma }}{t'}$.

%% file: exper.tex
\graphicspath{{../../}{./}}
\makeatletter
\def\input@path{{../../}{./}}
\makeatother

\begin{figure*}[t]
\begin{center}
\begin{subfigure}[b]{0.22\textwidth}
    \includegraphics[width=\textwidth]{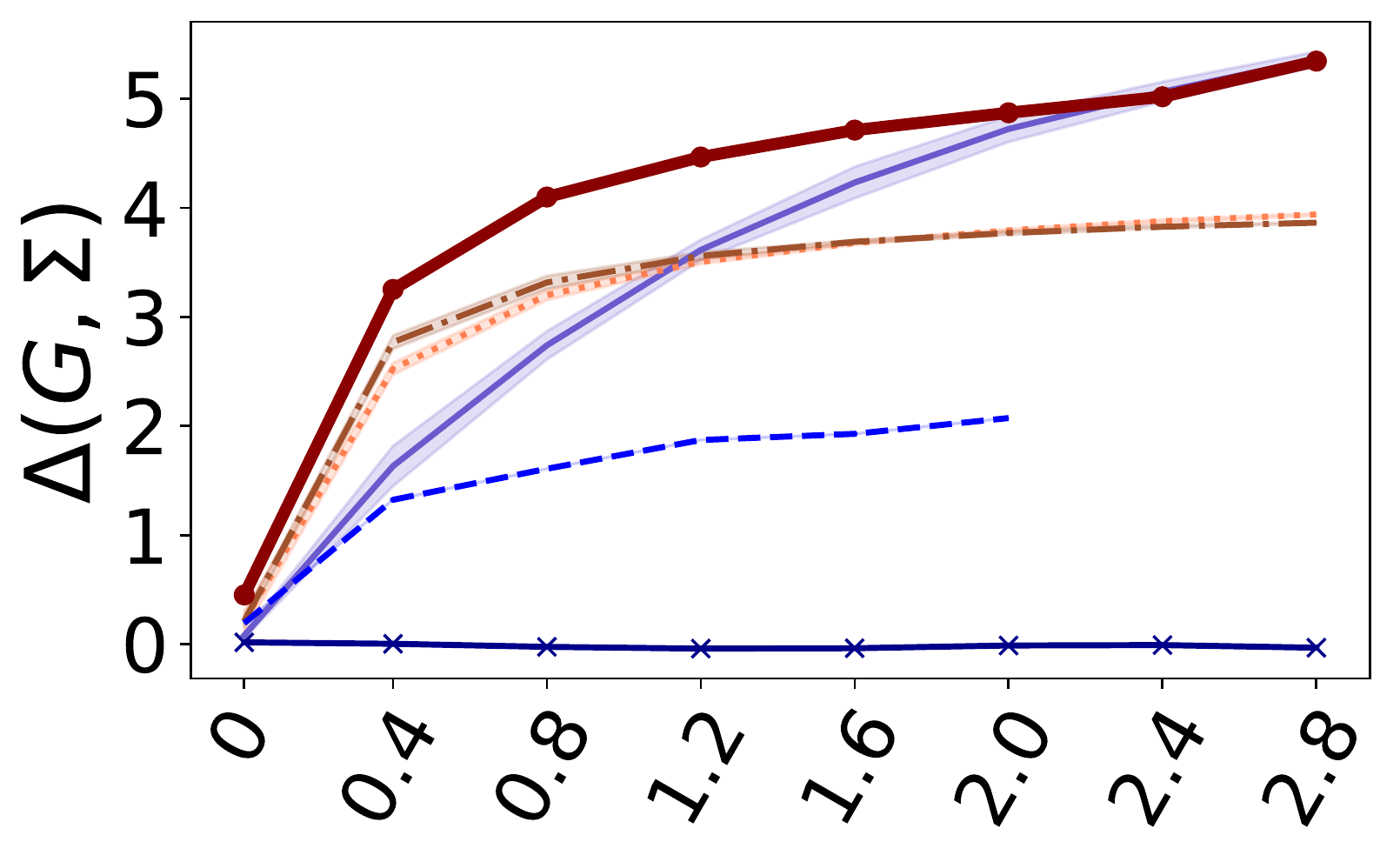}
    \includegraphics[width=\textwidth]{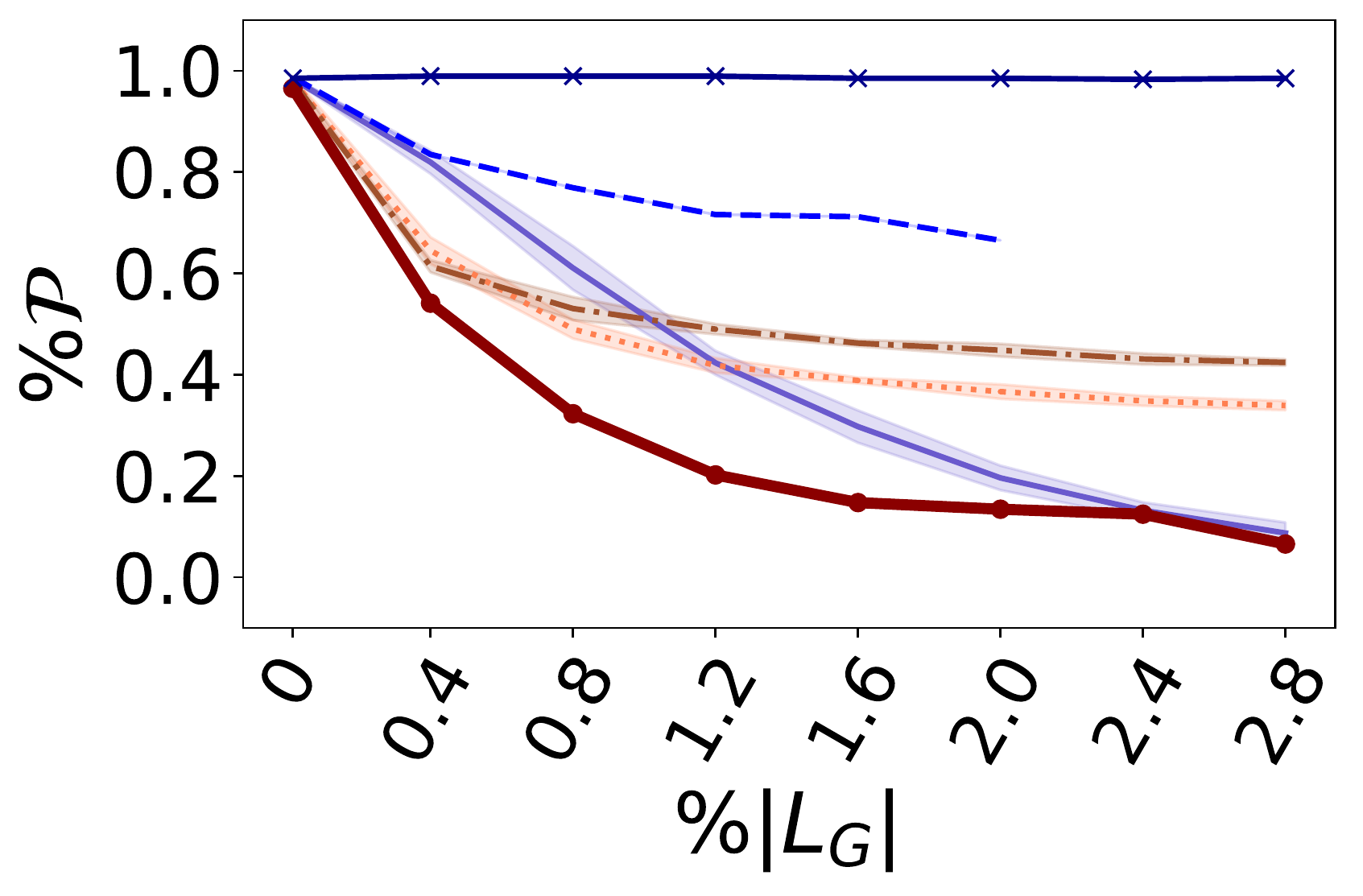}
    \ifextended%
    \else
    \fi
    \caption{Abortion}\label{fig:exp:algo:abortion:blue}
\end{subfigure}%
\begin{subfigure}[b]{0.22\textwidth}
    \includegraphics[width=\textwidth]{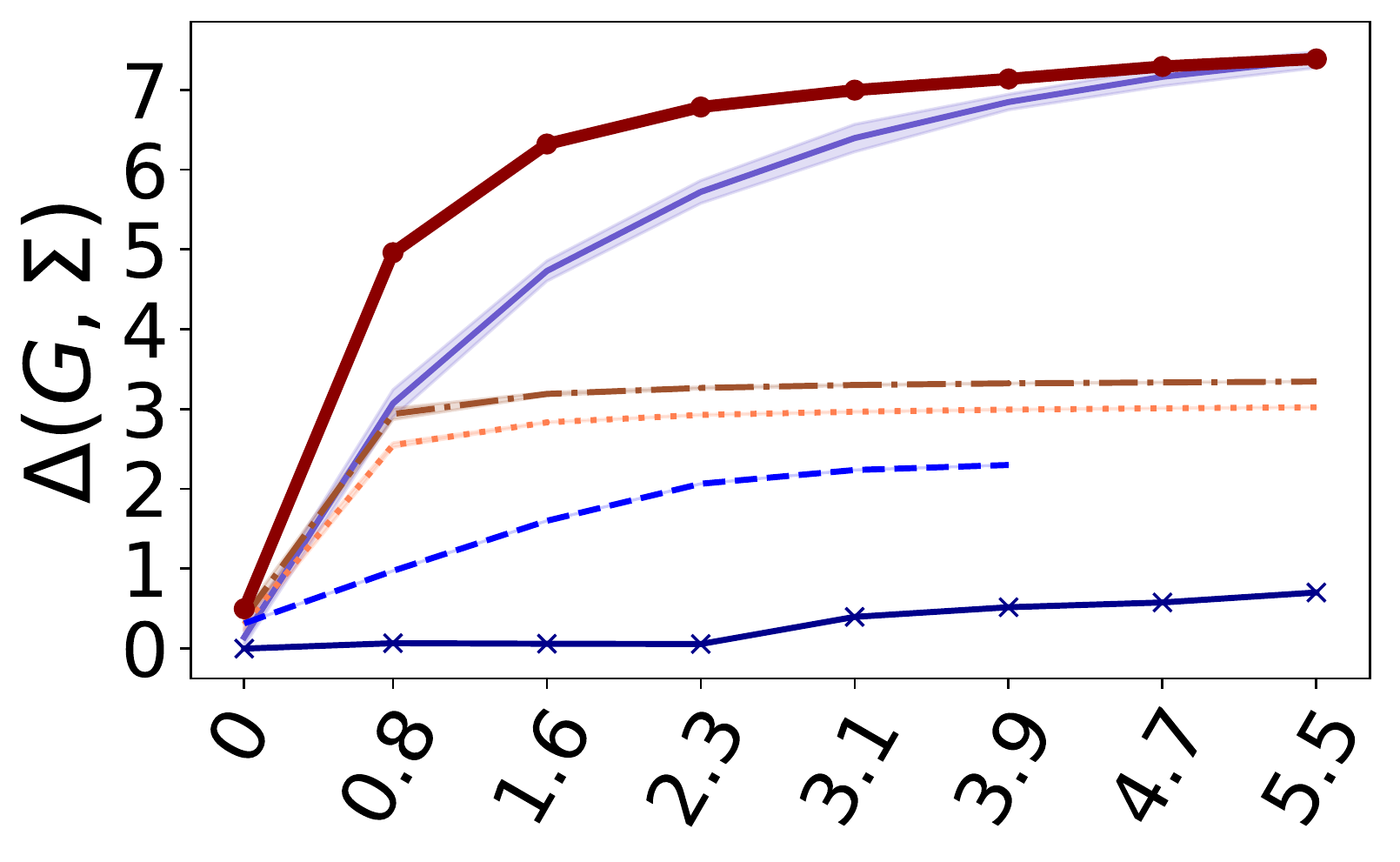}
    \includegraphics[width=\textwidth]{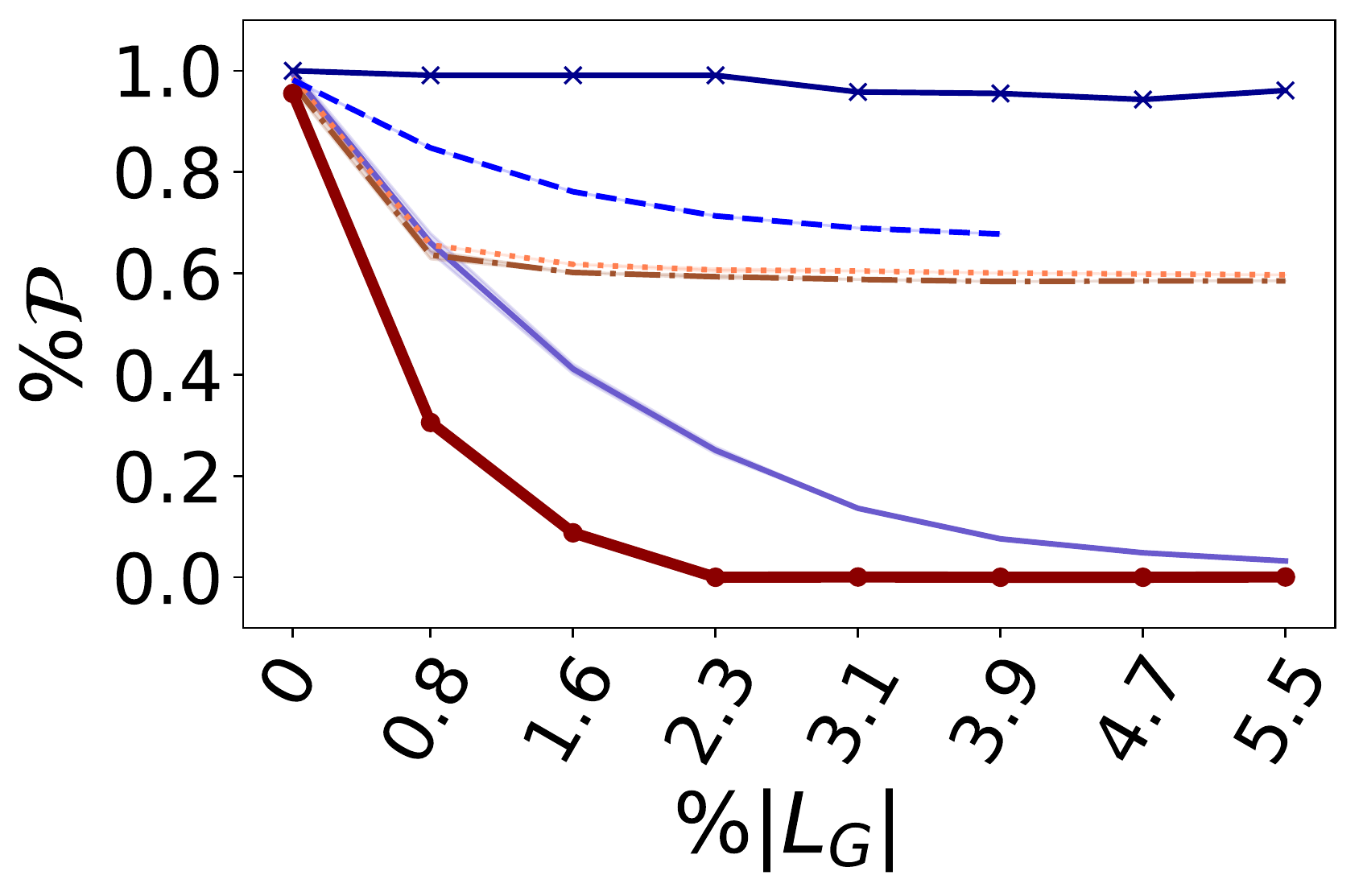}
    \ifextended%
    \else
    \fi
    \caption{MaTe}\label{fig:exp:algo:mate:blue}
\end{subfigure}
\begin{subfigure}[b]{0.22\textwidth}
    \includegraphics[width=\textwidth]{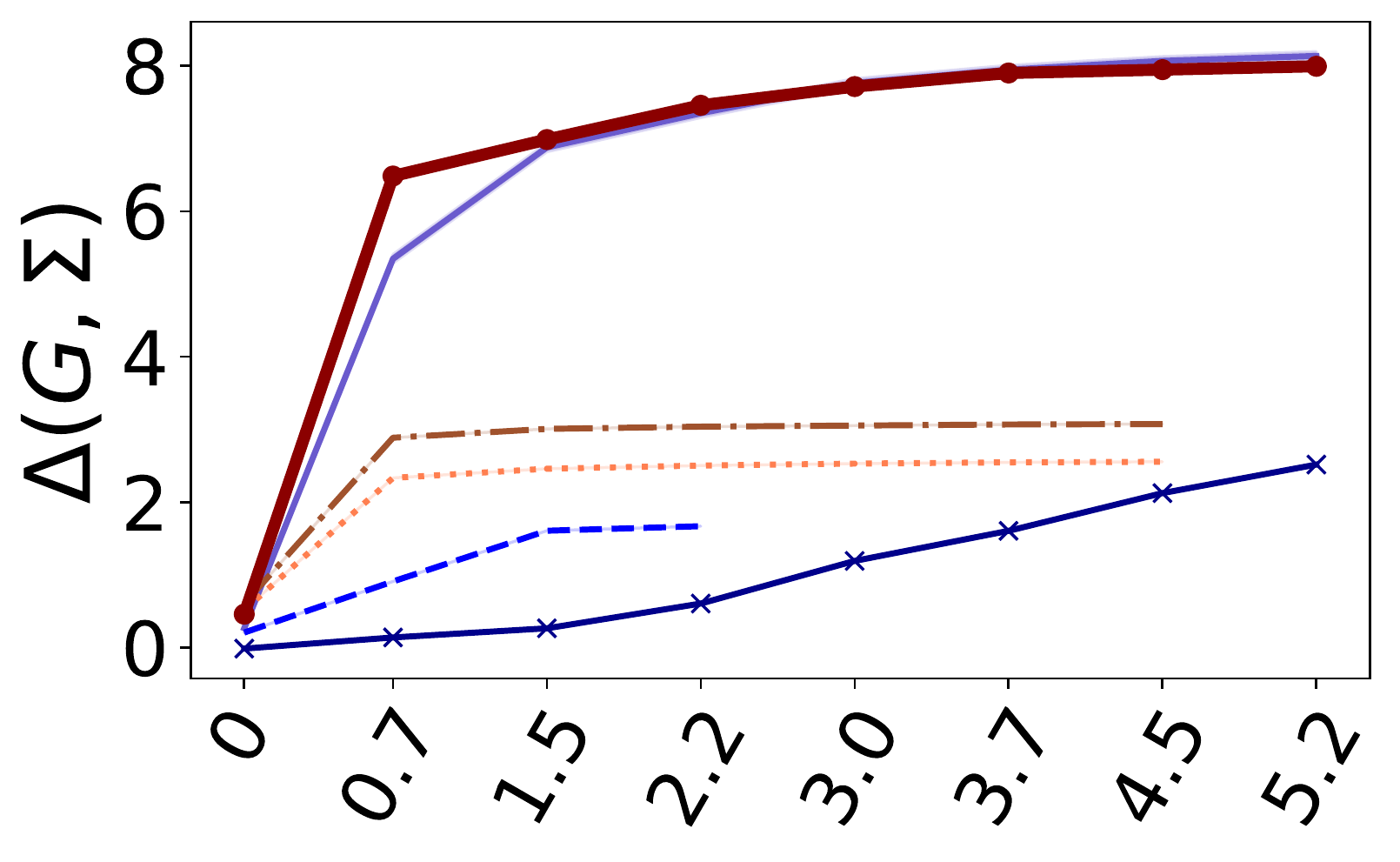}
    \includegraphics[width=\textwidth]{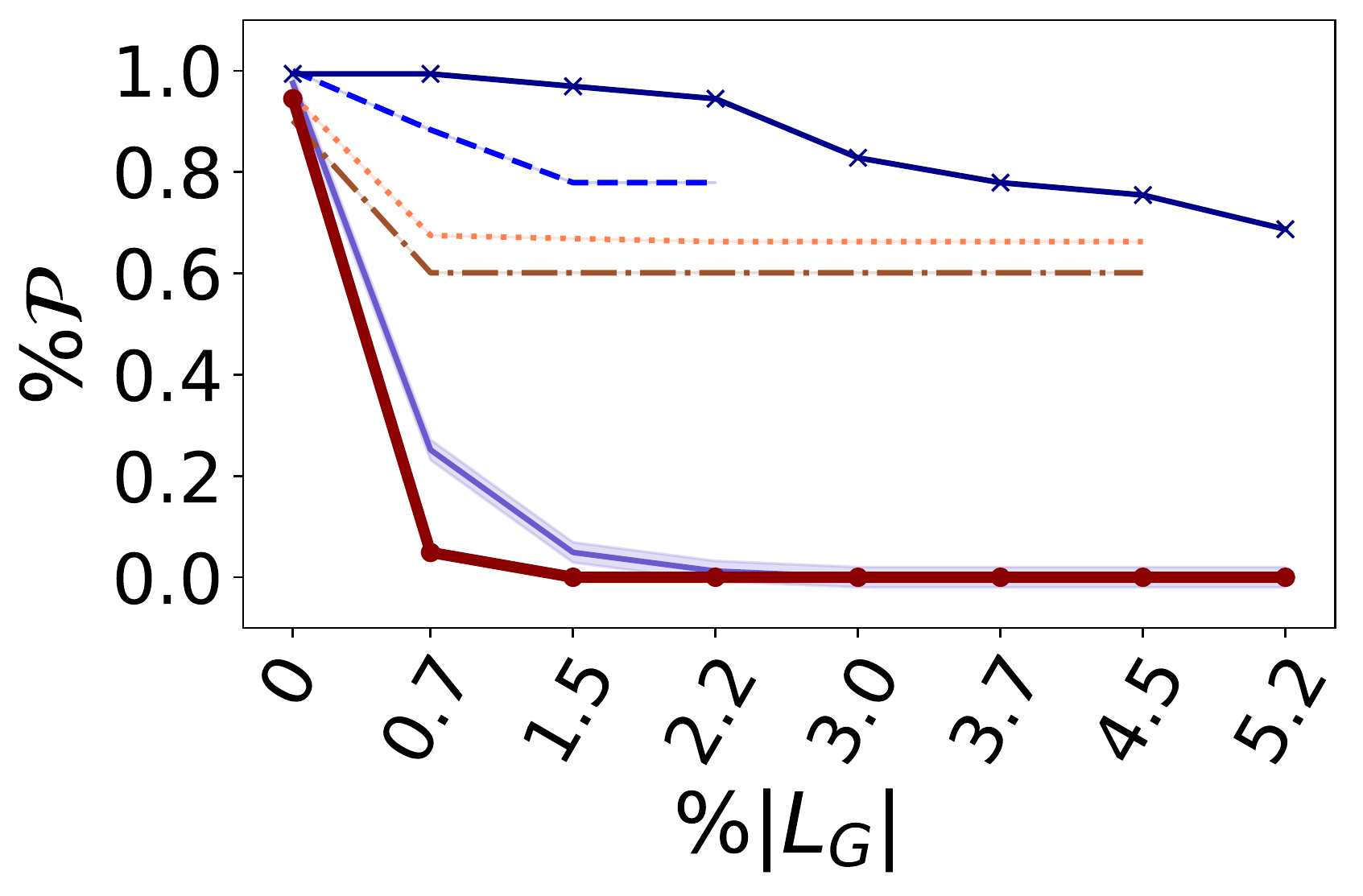}
    \ifextended%
    \else
    \fi
    \caption{MiHi}\label{fig:exp:algo:mihi:red}
\end{subfigure}
\begin{subfigure}[b]{0.22\textwidth}
    \includegraphics[width=\textwidth]{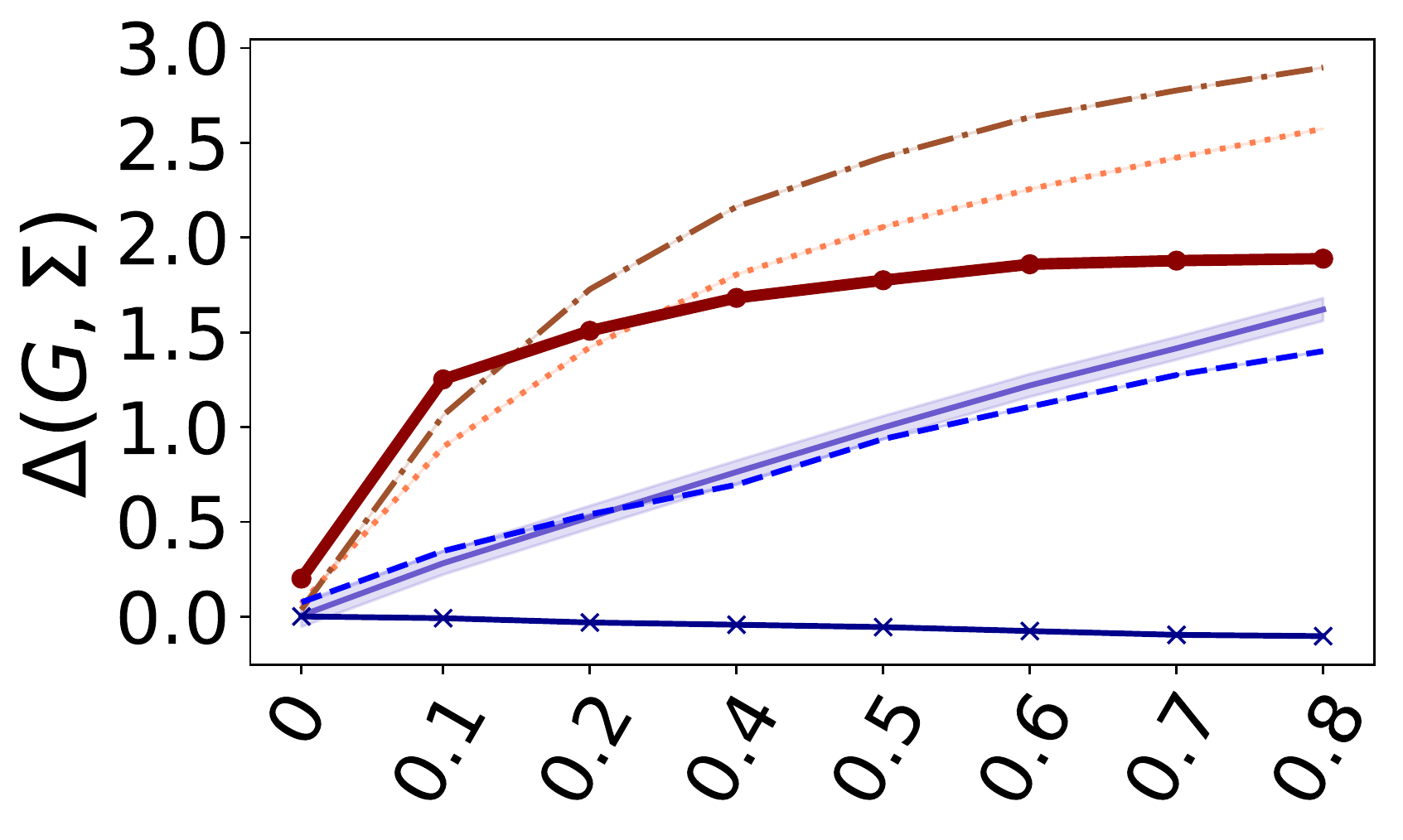}
    \includegraphics[width=\textwidth]{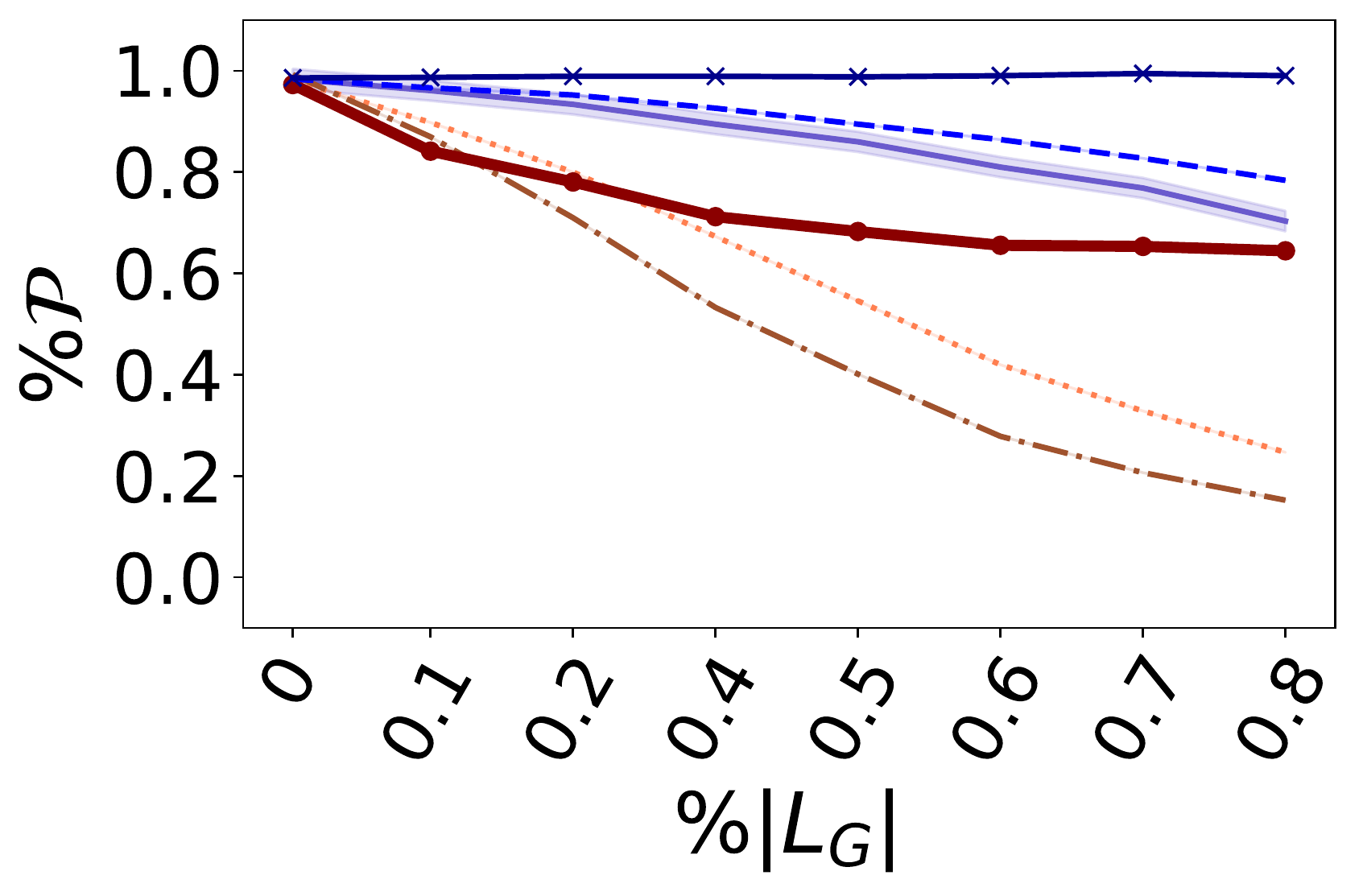}
    \ifextended%
    \else
    \fi
    \caption{PolBlogs}\label{fig:exp:algo:polblogs:blue}
\end{subfigure}
\begin{subfigure}[b]{0.63\textwidth}
    \includegraphics[width=\textwidth]{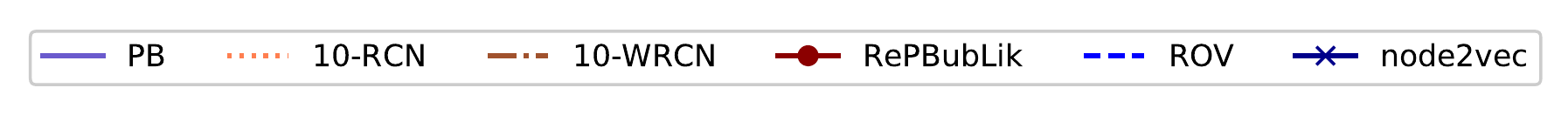}
\end{subfigure}
\end{center}
\ifextended%
\else
\fi
\caption{The first row shows the $\Delta(G, \Sigma)$ (y-axis) for increasing value of
$k$, reported in terms of $\%\mathcal{L}_G$, the union of possible edges across  $\badc{C}{G}$ and $\bar{C}$ for $C \in {R,B}$, (x-axis) for each algorithm. Higher values of
$\Delta$ show more significant reduction of the structural bias. In the second
row, we show the percentage of nodes that are still parochial, $\%\mathcal{P} =
\frac{\card{\pol{}{G}}-\card{\pol{}{\Gnew}}}{\card{\pol{}{G}}}$
after $k$ additions.}\label{fig:exp:total}
\Description{Results of the experiments for different graphs. See the caption
and the text for description of these results.}
\end{figure*}

\section{Experimental evaluation}\label{sec:exper}

The goal of our experimental evaluation is to understand how the addition of the
set $\Sigma = \Sigma_R \cup \Sigma_B$ of $K=k_R+k_B$ edges output by \algonameplus, run separately with $C=R$ and $B$, affects the structural bias of the
network, by computing the gain in the structural bias reduction. In particular, we measure the
gain with $\Delta(G, \Sigma)$, introduced in \cref{sec:algo}, used here with a
simpler notation.  We also measure the change
$\card{\bad{G}}-\card{\bad{\Gnew}}$  after adding $\Sigma$. 

\textit{Baselines.}
We compare \algonameplus\ to three different baselines (i.e., simplified variants of \algonameplus) and to two existing algorithms, described in the following. 
The first baseline, \textit{PureRandom} (PR) selects the source, and the target, nodes of the new edges uniformly at
random from the set $\badc{C}{G}$ and $\bar{C}$, respectively. The second baseline \textit{Random Top-$N$ Central Nodes
($N$-RCN)}, given a parameter $N\in (0,100)$, sorts the nodes in $\badc{C}{G}$
by descending centrality, and picks, uniformly at random, $k_C$ edges
with source in the top-$N$ percent of nodes in $\badc{C}{G}$. The last baseline,
\textit{Random Top-$N$ Weighted Central Nodes ($N$-RWCN)}, differs from $N$-RCN
as the nodes in $\badc{C}{G}$ are sorted in descending order by
$\centrdict(v) \times m_{v,u}$.


We compare \algonameplus\ also to two existing methods, ROV~\cite{garimella2017reducing}, and node2vec~\cite{grover2016node2vec}.  The \textit{ROV} algorithm outputs a set of $k$ edges to be added to $G$ to minimize the controversy score (RWC)~\cite{garimella2018quantifying}. The RWC is a metric that characterizes how controversial a topic is by capturing how well separated the two colors are. ROV considers as candidates the edges between the high-degree vertices of each color~\citep[Algorithm 1]{garimella2017reducing}. These edges are sorted by descending impact on the graph controversy score, and the top-$k$ edges are added to the graph. The objective of the comparison between ROV and \algonameplus\ is to verify whether an algorithm developed to minimize the RWC can be used to minimize the structural bias.
\textit{node2vec} is a graph embedding technique that encodes a network in a low-dimensional space retaining characteristics like the nodes' similarity~\citep{grover2016node2vec}. The generation of the embedding is based on random walks. One of the main applications of node2vec is to employ the embedding as the feature space to train link recommendation algorithms. The goal of comparing node2vec to \algonameplus\ is to understand how the predictions of widely-used link recommendation algorithms affect the network's structural bias. In the experiments, we create for each network a 128-dimensional space, then we train a logistic regression (\textit{avg.} AUC 85\%) over these features, and we predict the existence probabilities of edges from $\bad{G}$. We add to the graph the top $k$ edges according to these probabilities.

\textit{Datasets.}
We create graphs obtained from \textit{Wikipedia}
, \textit{Amazon}\footnote{\url{https://snap.stanford.edu/data/amazon-meta.html}} and
\textit{PolBlogs}\footnote{\url{http://www-personal.umich.edu/~mejn/netdata/}}. \Cref{tab:info} shows the relevant statistics.

From \textit{Wikipedia} we consider four bi-partitioned subgraphs related to controversial topics: \textit{politics, abortion, guns} and \textit{sociology}~\cite{menghini2020wikipedias}. 
Each node in the graph is a page, and is assigned to one color according to Wikipedia's categorization. Directed edges denote links, and are weighted using Wikipedia's clickstream data.\footnote{\url{https://dumps.wikimedia.org/other/clickstream/}}

The \textit{Amazon} dataset contains metadata about \textit{books}~\cite{leskovec2007dynamics}.
Given two book categories, the vertices are all the items in those categories, colored accordingly. There is a  directed edge $(u,v)$ if $v$ appears in the list of items similar to $u$. The edge is weighted by $v$'s sales rank.\footnote{Amazon sales rank is a metric of the relationship among products within one category based on their sales performance. It
expresses how well a product is selling relative to other products in the same category.}  We built three graphs by considering pairs of the following categories: \textit{Mathematics \& Technology (MaTe)}, \textit{ History of Technology \& Military Science (MiHi)}, and \textit{Mathematics \& Astronomy  (MaAs)}.

The \textit{Political Blogs} dataset is a directed network of hyperlinks between weblogs on US politics~\cite{adamic2005political}. Each node represents a blog and is colored according to its political leaning.
Links between blogs were automatically extracted from a crawl of the front page of the blog and represent the edges of the graph. Each edge $(v,u)$ has weight proportional to the out-degree of $v$.

\input{data_stats}

\textit{Setup.}
Given a network, we run \algoname\ and the other algorithms on that network for increasing values of $K$, with $K=1,2,4,6,\dotsc,400$ or $2000$ for larger graphs (\textit{Sociology} and \textit{Politics}). These values of $K$ represent only a small percentage of the set of possible edges to insert and correspond to the total number of edges to add to the graph. Once we set the value of K, accordingly, we allocate $k_B$ and $k_R$ of the $K$ edge insertions to each color proportionally to the sum of the BRs of the parochial vertices in each color.
In particular, we define $Y_C = \sum_{v \in \badc{C}{G}} \bubble{G}{v}{t}$, for $C \in {R,B}$, then $k_B = \left\lceil k \frac{Y_B}{Y_B+Y_R}\right\rceil$ and $k_R = K - k_B$. This allocation strategy is a simple but reasonable heuristic that ensures that more edges are added from nodes whose color is more parochial.

We assign the weight $m_{v,u}=1/(d(v)+1)$ to the added edge $(v,u)$, where $d(v)$ is the out-degree of $v$ before the insertion, and then we re-normalize the weights of the other edges by multiplying each of them by $1-m_{v,u}$. Furthermore, we set $\badthres=5$\ and $\goodthres=2$. Moreover, for the algorithms picking the top-N central nodes $N=10$. To account for variability of the algorithm, we run them 10 times. The variance of the results is low, overall.

The code for our experiments is available from \url{https://github.com/CriMenghini/RePBubLik}.

\textit{Experiment results.}
In \cref{fig:exp:total}, the plots in the first row show how the structural bias
is affected by the insertion of an incrementally larger set of edges, while the
ones on the second row show the reduction in the number of parochial nodes. Each
curve in the plot illustrates the gain by a different algorithm. We can draw the
following observations. (1) \algonameplus\ performs better than the
baselines and the competitors, especially after the insertion of a few edges, as
they obtain much larger gain with fewer insertions, i.e., the average BR of
parochial nodes decreases faster requiring less modifications. (2) N-RCN, N-WRC,
and ROV after a certain point become flat. (3) Overall, \algonameplus\  is the best
algorithm.
(4) The values of \algonameplus and PR converge, at different speed, to the same
value when we add more edges. (5) node2vec, in the best cases, shows little
improvement of the structural bias that, in the remaining cases, stays flat or
even increases. We now explain these behaviours using the plots on the second
row of \cref{fig:exp:total}.

(1) \algonameplus\ chooses edges that directly affect the BR of central nodes and,
with a chain effect, the BR of nodes connected to them. More central are the
nodes we attach the edges to, higher the structural bias drop is. In fact, it
follows, as shown for all the networks, that the addition of even small set of
edges is very effective. Additionally, we observe that the structural bias
reduction corresponds to a significant drop of the number of parochial nodes.

(2)  N-RCN, N-WRC, and ROV attach edges only to a subset of $\bad{G}$ and as $k$
increases, so does the probability of adding multiple edges to the same nodes.
These facts imply respectively that, especially on disconnected graphs (see MiHi
in \cref{fig:exp:algo:mihi:red}), the addition of edges may affect few nodes,
and that even the insertion of more edges does not modify the set of nodes on
which the new edges have effect. Thus, the curves of N-RCN, N-WRCN and ROV reach
an early saturation that expresses the scarce impact of subsequent edge
additions. This explanation is confirmed by the percentage of parochial nodes,
which does not decrease after the saturation point. Furthermore, the ROV shows a
stepping behaviour due to it selecting edges between high-degree central nodes
that minimize the RWC without imposing diversity constraints on nodes. And
resulting in many  selected edges being attached to the same node. Last, we see
that on \textit{Polblogs} the best algorithms are N-RCN, N-WRCN\@. This surprising
superiority of the random approaches can be explained by the fact that
\textit{Polblogs} is a connected graph, thus edges added to the top-central
nodes potentially affect all the nodes in $\bad{G}$. Thus, even when N-RCN and
N-WRCN add multiple new edges to the same set of nodes, $\Delta$ continues to
increase.

(3) \algonameplus\  shows a consistent behaviour, indeed it increases the gain faster than
other methods, requiring fewer insertions. The penalty factor $\eta$ allows the
algorithm to diversify the set of nodes to which the new edges attach, raising
the chances of lowering the BR of a larger number of parochial nodes, thus
increasing the gain. This feature is important especially on disconnected
graphs, where the vertices in tiny connected components always have lower
centrality compared to those in huge ones. More importantly, we observe that the
size of $\bad{G}$ is often reduced to 0: \algonameplus\  is able to ``heal'' all the bad
vertices, and if we measured the structural bias on the obtained graph it would
be zero.

(4) The variants of \algoname: \algonameplus\ and PR, pick edges from the same
candidate set, thus the more edges they can pick, the more likely they choose
edges with similar effect, thus the average parochial nodes' BR converges. This is the main explanation why the random algorithm performs so well.

(5) Generally, link recommendation algorithms tend to suggest edges between
similar nodes. Node2vec captures this similarity through the nodes'
neighborhood. In this context, graphs partitions have high within- and low
between-density. Nodes in the same partition then lie close in the
embedding space. Edges suggested by node2vec with high probability connect nodes
close to each other in the embedding, which often are in the same partition.
Thus, node2vec has a hard time reducing the structural bias, and in some cases
increases it.

\ifextended%
Plots for \textit{guns}, \textit{sociology}, \textit{politics} and \textit{MaAs}
show similar behaviour and can be found in \cref{sec:appendix}.
\else
Due to space constraint we omit the presentation of the plots for \textit{guns},
\textit{sociology}, \textit{politics} and \textit{MaAs}, which show similar
behaviour. But they can be found in the extend online version~\citemissing.
\fi


%% file: data_stats.tex
\begin{table}[ht]
\centering
\resizebox{\columnwidth}{!}{
\begin{tabular}{lccccccc}  
\toprule
\multicolumn{8}{c}{\textit{Wikipedia}} \\
\cmidrule(r){1-8}
Topic    & $\card{R}$ & $\card{B}$  &  $\card{E}_{R \rightarrow B}$ &
$\card{E}_{B \rightarrow R}$ & $\card{E}$  & \%$\badc{R}{G}$ &  \%$\badc{B}{G}$\\
\midrule
\textit{Abort.}       & 208     & 413  &  80  & 170 & 1911 & 85.56 & 89.20\\
\textit{Guns}       & 142     & 118  &  72  & 79 & 723  & 82.95 & 71.69\\
\textit{Pol.}      & 10347    & 10129  &  17452  & 16484 & 141486  & 25.97 & 42.36\\
\textit{Sociol.}      & 602   &  2283 & 284  & 192 & 10514 &  91.32 & 96.36\\
\midrule
\multicolumn{8}{c}{\textit{Amazon}} \\
\cmidrule(r){1-8}
Topic    & $\card{R}$ & $\card{B}$  &  $\card{E}_{R \rightarrow B}$ &
$\card{E}_{B \rightarrow R}$ & $\card{E}$  & \%$\badc{R}{G}$ &
\%$\badc{B}{G}$\\
\midrule
\textit{MaTe}       & 827     & 566  &  25  & 42 & 675 &  90.91 & 79.63\\
\textit{MiHi}       & 446     & 405  &  66  & 63 & 482 &  58.33 & 63.46\\
\textit{MaAs}      & 827   &  294 & 11  & 6 & 680 &  97.31 & 95.15\\
\midrule
\multicolumn{8}{c}{\textit{PolBlogs}} \\
\cmidrule(r){1-8}
Topic   & $\card{R}$ & $\card{B}$  &  $\card{E}_{R \rightarrow B}$ &
$\card{E}_{B \rightarrow R}$ & $\card{E}$  & \%$\badc{R}{G}$ &
\%$\badc{B}{G}$\\
\midrule
\textit{Politics}      & 545    & 488  &  902  & 781 & 17348 &  87.71 & 90.37\\
\bottomrule
\end{tabular}}
\caption{Networks' statistics. The notation is consistent with the rest of the
paper.}\label{tab:info}
\end{table}

%% file: concl.tex
\section{Conclusion}\label{sec:concl}
We presented \algoname, an algorithm that reduces the structural bias of a graph
by adding $k$ edges. Thanks to the monotonicity and submodularity of the
objective function, \algoname\ is able to return a constant-factor approximation
using a greedy approach based on a task-specific variant of the random walk
closeness centrality. The results of our experimental evaluation show that the
edge insertions suggested by \algoname\ result in a much quicker decrease of the
structural bias than existing methods and reasonable baselines.

The functionality of \algoname \ relies on the existence of an oracle receiving the network and a page in it as input and outputting the transition probabilities of potentially added links to the input page. 
We leave the question of designing an algorithm which learns such probabilities from data as future direction of this work.

%% file: appendix.tex
\section{Missing proofs}\label{sec:appendix}
We present here the proofs missing from the main body. For convenience, we
repeat the statements of the lemmas.

\lemcentr*

\begin{proof}[Proof of \cref{lem:centr}]
  We can write
  \[
    \centr{t'}{v}{S} = t' - \frac{1}{\card{S}} \sum_{ w \in S }
    \expect{G}{\ttime{t'}{w}{v}} \enspace.
  \]
  We apply Chebyshev's inequality to the r.v.\ $\nicefrac{1}{z} \sum_{i=1}^z
  \bar{h}_{w_{i}}$, to bound the deviation from its expectation
  \[
    \frac{1}{\card{S}} \sum_{ w \in S } \expect{G}{\ttime{t'}{w}{v}} \enspace.
  \]
  To get an upper bound to the variance of this r.v., we use the fact that the
  r.v.'s $\bar{h}_{w_i}$, $i=1,\dotsc,z$, are independent, and, from Popoviciu's
  inequality, the fact that each has a variance at most $\nicefrac{t'^2}{4}$, as
  $\bar{h}_{w_i} \in [0,t']$.
\end{proof}

The following result is used in the proof of \cref{lem:restart}.

\begin{lemma}[Markov inequality for bounded random variables]\label{lem:markov}
  Let $X$ be a random variable satisfying $0 \leq X \leq t$. We have:
  \[
    \cP(X \leq k) \leq \frac{t - \expect{}{X}}{t - k} \enspace.
  \]
\end{lemma}

\begin{proof}
  It holds
  \begin{align*}
    \expect{}{X} &= \int_{0}^{k} x p(x) dx  +\int_{k}^{t} x p(x) dx \\
    &\leq k \left( 1 - \cP(X \geq k) \right) + t \cP(X \geq k) \enspace.
  \end{align*}
  Thus,
  \[
    \cP(X \geq k) \geq \frac{\expect{}{X} - k}{t - k},
  \]
  and
  \[
    \cP(X \leq k) = 1 - \cP(X \geq k) \leq 1 - \frac{\expect{}{X} - k}{t - k} =
    \frac{t - \expect{}{X}}{t - k} \enspace. \qedhere
  \]
\end{proof}

\lemrestart*

\begin{proof}[Proof of \cref{lem:restart}]
  Assume first that $\bubble{G}{v}{t} \geq t (1-\nicefrac{1}{8r})$. Consider a
  set of $r$ independent random walkers, $w_1, \dotsc, w_r$, each starting from
  $v$. We can see the trace of the partial walks taken by our random walker with
  restarts as the union of the traces of these walkers. The event $\mathcal{E}'
  \doteq \text{``}\mathcal{T}_v \le \nicefrac{t}{2}\text{''}$ is a strict subset
  of the event $\mathcal{E}'' \doteq \text{``there is (at least) a walker } w_i$
  for which $T_v^t \le \nicefrac{t}{2}$'', as the condition in $\mathcal{E}'$
  implies the condition in $\mathcal{E}''$, but not vice versa. Thus,
  $\cP(\mathcal{E}') < \cP(\mathcal{E}'')$. By  \cref{lem:markov} we have, for
  each walker, that
  \[
    \cP\left(T^t_v\leq \frac{t}{2} \right) \leq
    \frac{t - \expect{}{\ttime{t}{v}{S}}}{t - \frac{t}{2}} \le
    \frac{\frac{t}{8r}}{\frac{t}{2}} \leq \frac{1}{4r}.
  \]
  Thus, using the union bound over the $r$ walkers, we get $\Pr({\mathcal
  E}'') \le \nicefrac{1}{4}$. Equivalently $\cP\left(\mathcal{T}_v\leq \nicefrac{t}{2}
  \right) \leq \nicefrac{1}{4}$.

  For the case when $\bubble{G}{v}{t} \leq b$,
  using Markov inequality we get $\cP\left(\mathcal{T}_v > 4br \right) \leq
  \nicefrac{1}{4}$.
\end{proof}

\lemreduction*

\begin{proof}[Proof of \cref{lem:reduction}]
  We show an approximation-preserving polynomial time  reduction from the
  minimum set cover problem to \cref{prob:zerobias}. Our reduction  does
  \emph{not} change the cost of the optimal solution, thus maintaining, in
  addition to NP-hardness, the APX-hardness.

  Let $U=\{u_1, u_2,\dots , u_n\}$ be a domain and let $S_1,S_2,\dotsc, S_m$\\
  $\subseteq U$ be an instance of the set cover problem. We construct an instance
  of \cref{prob:zerobias} as follows. Fix $t \ge 3$. Let $V$ be union of the
  following sets: $U$,  $S={\{s_i\}}_{i=1}^m$ representing the sets, $T =
  \bigcup_{j=1}^m T_j$ where each $T_j$ is a set of $\lceil t/2 \rceil -1 $
  distinct vertices, and $\{g\}$. Assume all vertices except $g$ have color red
  and $g$ is blue. For each $i\in [n]$ and $j \in [m]$, place an edge from $u_i$
  to $s_j$ if and only if $u_i\in S_j$. For each $j\in [m]$, using the vertices
  in $T_j$, place a path of length $\lceil t / 2\rceil - 1$  going from $s_j$ to
  $g$. For each $1\leq j\leq m$, it holds $\bubble{G}{s_j}{t} = \lceil
  \nicefrac{t}{2} \rceil - 1$, and for each $1\leq i\leq n$,
  \[
    \bubble{G}{u_i}{t} = \frac{1}{\card{\{j\ :\ u_i \in S_j\}}} \sum_{j\ 
    \text{s.t.}\ u_i\in S_j} \bubble{G}{s_j}{t}+1 = \lceil \nicefrac{t}{2}
    \rceil \enspace.
  \]
  Clearly the bubble radius of vertices in $T$ is strictly less than
  $\nicefrac{t}{2}$. Thus the parochial vertices are all and only those in $U$.
  Assume there is a polynomial-time algorithm for \cref{prob:zerobias}. For any
  (optimal) solution $\Sigma \subseteq V \times V$, it holds
  $\bubble{G_{\textrm{new}}}{u_i}{t} < \nicefrac{t}{2}$ if and only if $\Sigma$
  contains an edge whose source is in $\{u_i\} \cup
  \bigcup_{j\ \text{s.t.}\ u_i \in S_j} (\{s_j\} \cup T_j)$, for each $i \in
  [n]$. The source vertices of the edges in $\Sigma$ must be distinct, as any
  solution containing two edges originating from the same vertex cannot be
  optimal. Denote with $Z$ the set of the source vertices of the edges in
  $\Sigma$. Consider now the solution $\Sigma'$ obtained by changing (in
  polynomial time) $\Sigma$ as follows: \textit{1.} each edge in $\Sigma$ whose
  source is in $T_i$ is modified to have source $s_i$, for each $i \in [n]$; and
  \textit{2.}  each edge in $\Sigma$ whose source is $u \in U$ is changed to
  have source $s_j$ where $j$ is such that $u \in S_j$. Clearly $\Sigma'$ is
  still an (optimal) solution to \cref{prob:zerobias}. Let $\mathsf{OPT}$ be the
  set of source vertices of the edges in $\Sigma'$. Clearly it must be
  $\mathsf{OPT} \subseteq S$. We now show that $\Sigma'$ is an (optimal)
  solution to \cref{prob:zerobias} if and only if $\mathsf{OPT}$ is such
  that $\{ S_j \:\ s_j \in \mathsf{OPT} \}$ is a minimum set cover for the
  considered instance. It is evident that $\{ S_j\ :\ s_j \in \mathsf{OPT} \}$
  is a set cover, which can be obtained in polynomial time from $\Sigma'$. We
  now show that this set cover is minimal. Consider now any set cover $Y
  \subseteq \{S_1,\dotsc, S_m\}$, and consider the set of edges $\{(s_i, g)\ :\
  S_i \in Y\}$. Adding these edges to $G$ would result in all the vertices in
  $U$ to no longer be parochial. This holds in particular for any \emph{minimal}
  set cover  $Y$, from which we can create an (optimal) solution $\Sigma_Y$ to
  \cref{prob:zerobias}. Thus we found a bijection between (optimal)
  solutions to \cref{prob:zerobias} and minimal set covers for the
  considered instance, and computing one from the other can be done in
  polynomial time, showing the NP-hardness of \cref{prob:zerobias}. The
  APX-hardness follows because, for any minimum set cover $Y$, the corresponding
  optimal solution $\Sigma_Y$ to \cref{prob:zerobias}, built as above, is
  such that $\card{\Sigma_Y} = \card{Y}$, thus if we had a constant-factor
  polynomial-time approximation algorithm for \cref{prob:zerobias} we
  would have an algorithm with the same properties for the minimum set cover
  problem.
\end{proof}

\lemgainbounds*

\begin{proof}[Proof of \cref{lem:gainbounds}]

  Consider the probability space of all random walks starting from $v$ in $\Gnew$ and $G$.  We introduce a coupling between these two probability spaces as follows:  consider a walk in $\Gnew$ and couple every step of it to an identical step in $G$. 
  If a walk in $\Gnew$ never traverses $(v,w)$ then the gain function is zero as it gets coupled to the identical walk in $G$. 
  Assume that the walk in $\Gnew$ traverses $(v,w)$ at the $i$th step without first visiting a vertex in $\bar{C}_v$. Before traversing $(v,w)$, the two identical walks in $\Gnew$ and $G$ have the same probabilities and the above coupling works. We partition the state space by conditioning on the step $i$ as follows:

  Let $\mathcal{E}_i$, $1 \le i \le t'$,  be  the  event  that the walk in $\Gnew$ traverses $(v,w)$ at step $i$. Consider all such walks,
  at step $i-1$ these walks  need one more steps to reach the other color, and they are coupled to walks in $G$ which in expectation need $\bubble{G}{v}{t'-i+1}$ steps to reach $\bar{C}_v$. Thus, assuming $\mathcal{E}_i$,  the gain in bubble radius is equal to $\bubble{G}{v}{t'-i+1}-1$.

 Using the law of total expectation and summing over all $1\leq i\leq t'$, we can write
     \begin{equation}\label{eq:4}
    \bubblechange{G}{v}{(v,w)}{m_{vw}}{t'} = \sum_{i=1}^{t'} \left(
    \bubble{G}{v}{t'-i+1} - 1 \right) \cP(\mathcal{E}_i) \enspace.
    \end{equation} 
The left hand side follows from the fact that $\mathbb{P}({\mathcal E}_1)=m_{vw}$ and that $\bubble{G}{v}{j} \ge 1$ for any $1 \le j \le t'$. 
 The right-hand side is concluded from the fact that $\bubble{G}{v}{t'-i+1}\leq\bubble{G}{v}{t'} $ and that
  \[
    \sum_{i=1}^{t'} \cP (\mathcal{E}_i) = \sum_{i=0}^{t'-1} \visitattime{v}{i}  m_{vw}= \mathcal{F}_{t'}(v)m_{vw} \enspace.
  \]
 \end{proof}

\lemallvertices*

\begin{proof}[Proof of \cref{lem:allvertices}]

  Using the law of total expectation, for any graph $Z$, it holds
  \begin{align*}
    \bubble{Z}{u}{t} =& \left( \sum_{i=1}^{t-1} \left( i + \bubble{Z}{v}{t-i}
    \right) \cP\left( \towithin{u}{v}{=i}{Z} \right) \right) \\
    & + \expect{Z}{\ttime{t}{u}{\bar{C}_v} \mid \nottowithin{u}{v}{< t}{Z}}
    \cP\left( \nottowithin{u}{v}{< t}{Z} \right) \enspace.
  \end{align*}
  Between $G$ and $\Gnew$, we are only adding an outgoing edge from $v$ and
  modifying the weights of the edges outgoing from $v$, so
  \begin{align*}
    \expect{G}{\ttime{t}{u}{\bar{C}_v} \mid \nottowithin{u}{v}{< t}{G}} =
    \expect{\Gnew}{\ttime{t}{u}{\bar{C}_v} \mid \nottowithin{u}{v}{<
    t}{\Gnew}},\\
    \cP\left(\nottowithin{u}{v}{< t}{G} \right) = \cP\left(\nottowithin{u}{v}{<
    t}{\Gnew} \right),\ \text{and}\ \cP\left( \towithin{u}{v}{=i}{G} \right) = \cP\left(
    \towithin{u}{v}{=i}{\Gnew} \right) \enspace.
  \end{align*}

  Therefore,
\begin{align*}
  \bubblechange{G}{u}{(v,w)}{m_{v}}{t} &\doteq
  \bubble{G}{u}{t}-\bubble{G_{\textrm{new}}}{u}{t}\\
   &= \sum_{i=1}^{t-1} \left( \bubblechange{G}{v}{(v,w)}{m_{v}}{t-i}\right)
   \cP\left( \towithin{u}{v}{=i}{G} \right) \\
   &= \sum_{i=1}^{t-2} \left( \bubblechange{G}{v}{(v,w)}{m_{v}}{t-i}\right)
   \cP\left( \towithin{u}{v}{=i}{G} \right)
   \enspace.
 \end{align*}
 The last step follows from the fact that $\bubblechange{G}{v}{(v,w)}{m_v}{1}
 =0$ because $\bubble{Z}{u}{1}=1$ for every vertex $u$ of any graph $Z$.
\end{proof}

We need the following technical result in successive proofs.

\begin{lemma}\label{lem:parochiallowert}
  If $\bubble{G}{v}{t} \geq \nicefrac{t}{2}$ then, for any
  $t' \leq t$, it holds
  \[
    \frac{t'}{2} \leq \bubble{G}{v}{t'} \leq t' \enspace.
  \]
\end{lemma}

\begin{proof}
  The rightmost inequality is straightforward from the definition of
  $\bubble{G}{v}{t'}$. Expanding the definition of $\bubble{G}{v}{t}$,
  it holds, for any $t' < t$,
  \begin{align*}
    \bubble{G}{v}{t} &= t \cP \left( \towithin{v}{\bar{C}_v}{\ge t}{G} \right)
    + \sum_{i=1}^{t-1} i \cP \left( \towithin{v}{\bar{C}_v}{= i}{G} \right) \\
    &\leq  t \cP \left( \towithin{v}{\bar{C}_v}{\ge t}{G} \right)
    + t \cP \left( \towithin{v}{\bar{C}_v}{t' \le \cdot \le t}{G} \right)
    + \sum_{i=1}^{t'-1} i \cP \left( \towithin{v}{\bar{C}_v}{= i}{G} \right) \\
    & = t \cP \left( \towithin{v}{\bar{C}_v}{\ge t'}{G} \right)
    + \sum_{i=1}^{t'-1} i \cP \left( \towithin{v}{\bar{C}_v}{= i}{G} \right)
    \enspace.
  \end{align*}
  Thus, since the l.h.s.~is at least $\nicefrac{t}{2}$, 
  \[
    \frac{t}{2} \leq t \cP \left( \towithin{v}{\bar{C}_v}{\ge t'}{G} \right)
    + \sum_{i=1}^{t'-1} i \cP \left( \towithin{v}{\bar{C}_v}{= i}{G} \right),
  \]
  i.e.,
  \[
    \frac{1}{2} - \frac{1}{t} \sum_{i=1}^{t'-1} i \cP \left(
    \towithin{v}{\bar{C}_v}{= i}{G} \right)
    \leq \cP \left( \towithin{v}{\bar{C}_v}{\ge t'}{G} \right)
    \enspace.
  \]
  By expanding $\bubble{G}{v}{t'}$ in a similar way, and plugging in the last
  inequality above, we get
  \begin{align*}
    \bubble{G}{v}{t'} &= t' \cP \left( \towithin{v}{\bar{C}_v}{\ge t'}{G} \right)
    + \sum_{i=1}^{t'-1} i \cP \left( \towithin{v}{\bar{C}_v}{= i}{G} \right)\\
    &\geq  t' \left( \frac{1}{2} - \frac{1}{t} \sum_{i=1}^{t'-1} i
    \cP \left( \towithin{v}{\bar{C}_v}{= i}{G} \right)
    \right) + \sum_{i=1}^{t'-1} i \cP \left( \towithin{v}{\bar{C}_v}{= i}{G}
    \right)\\
    & = \frac{t'}{2} + \underbracket{\left( 1 - \frac{t'}{t} \right)
      \sum_{i=1}^{t'-1} i \cP \left( \towithin{v}{\bar{C}_v}{= i}{G}
    \right)}_{\ge 0}\\
    &\geq \frac{t'}{2} \enspace. \qedhere
  \end{align*}
\end{proof}

\lemgain*

\begin{proof}[Proof of \cref{lem:gain}]
  Using \cref{lem:allvertices}, we get
  \begin{align}
    &\bubblechange{G}{\pol{C_v}{G}}{e}{m_e}{t} = \nonumber\\
    &\frac{1}{\card{\pol{C_v}{G}}}
    \sum_{u \in \pol{C_v}{G}} \sum_{i=1}^{t-2}
    \bubblechange{G}{v}{e}{m_e}{t-i}
    \cP\left( \towithin{u}{v}{=i}{G} \right) \enspace.\label{eq:gaintech}
  \end{align}
  It holds from \cref{lem:gainbounds,lem:parochiallowert} that
  \[
    \bubblechange{G}{v}{e}{m_e}{t'} \geq \left( \frac{t'}{2} - 1 \right)
    m_e\ \text{for every}\ 1\leq t' \le t\enspace.
  \]
  Using this fact, we can continue from~\eqref{eq:gaintech} as follows
  \begin{align*}
    &\bubblechange{G}{\pol{C_v}{G}}{e}{m_e}{t} \\
    &\ge \frac{1}{\card{\pol{C_v}{G}}} \sum_{u \in \pol{C_v}{G}} \sum_{i=1}^{t-2}
    \left( \frac{t-i}{2}-1 \right)  m_e
    \cP\left( \towithin{u}{v}{=i}{G} \right)\\
    &= \frac{m_e}{2}  \underbracket{\frac{1}{\card{\pol{C_v}{G}}} \sum_{u\in \pol{C_v}{G}}
    \sum_{i=1}^{t-2} (t-i-2) \cP\left( \towithin{u}{v}{=i}{G}
  \right)}_{\centr{t-2}{v}{\pol{C_v}{G}}},
  \end{align*}
  which concludes the proof.
\end{proof}

\lemopt*

\begin{proof}[Proof of \cref{lem:opt}]
  It follows from \cref{lem:gainbounds} that, for any $t'$,
 \begin{align*}
   \bubblechange{G}{\optnode}{e_\optnode}{m_\optnode}{t'}
   &\le \mathcal{F}_{t'}(u) \left( \bubble{G}{\optnode}{t'}
   - 1\right) m_\optnode \\
   &\leq (t'- 1) m_\optnode \mathcal{F}_{t'}(\optnode) \enspace.
 \end{align*}
 By applying \cref{lem:allvertices} first, and then the above inequality, we get
 \begin{align*}
  &\bubblechange{G}{\pol{C}{G}}{e_\optnode}{m_\optnode}{t}\\
  \le& \frac{1}{\card{\pol{C}{G}}}\sum_{u \in \pol{C}{G}} \sum_{i=1}^{t-2}
  \left( \bubblechange{G}{\optnode}{e_\optnode}{m_\optnode}{t-i} \right)
  \cP \left( \towithin{u}{v}{= i}{G} \right)\\
    \le& \frac{1}{\card{\pol{C}{G}}}\sum_{u \in \pol{C}{G}} \sum_{i=1}^{t-3}
  \left( \bubblechange{G}{\optnode}{e_\optnode}{m_\optnode}{t-i} \right)
  \cP \left( \towithin{u}{v}{= i}{G} \right)\\
  &\hspace{2cm}+
  \bubblechange{G}{\optnode}{e_\optnode}{m_\optnode}{2} 
  \cP \left(\towithin{u}{v}{= t-2}{G}\right)\\
      \le& \frac{1}{\card{\pol{C}{G}}}\sum_{u\in\pol{C}{G}} \sum_{i=1}^{t-3}
      (t - 1 - i) m_\optnode \mathcal{F}_t(\optnode)
   \cP \left( \towithin{u}{v}{= i}{G} \right) +1\\
   \le& \frac{1}{\card{\pol{C}{G}}}\sum_{u \in \pol{C}{G}} \sum_{i=1}^{t-2}
   2 (t - 2 - i) m_\optnode \mathcal{F}_t(\optnode)
  \cP \left( \towithin{u}{v}{= i}{G} \right)+1 \\
   \le&\ 2 m_\optnode
   \centr{t-2}{\optnode}{\bar{C}_\optnode}  \mathcal{F}_t(\optnode)+1 \le
   2 m_v \centr{t-2}{v}{\bar{C}_\optnode} \gamma(G)+1\\
   \le&\ ( 4 \gamma(G)+1) \bubblechange{G}{\pol{C}{G}}{e_v}{m_v}{t},
   \end{align*}
   where the last step follows from \cref{lem:gain}.
\end{proof}

\lemsubmodular*

\begin{proof}[Proof of \cref{lem:submodular}]
  Let $G_v$ be the graph after adding only the edge $e_v$, $G_u$ be the graph
  after only adding the edge $e_u$, and $G_{vu}$ be the graph after adding both
  edges.

  We first show the monotonicity of the objective function, i.e.,
  that~\eqref{eq:monotonicity} holds. For any $w \in \pol{C}{G}$, it holds
  \begin{align*}
    \bubblechange{G}{w}{e_v}{m_v}{t} &\doteq \bubble{G}{w}{t} -
    \bubble{G_v}{w}{t} \\
    &\le \bubble{G}{w}{t} - \bubble{G_{vu}}{w}{t} \\
    &\doteq \bubblechange{G}{w}{\{e_v, e_u\}}{\{m_v, m_u\}}{t}
  \end{align*}
  because $\bubble{G_v}{w}{t} \geq \bubble{G_{vu}}{w}{t}$, as adding an edge
  from $u$, which is in $C_w$, to a vertex in $\bar{C}_w$ cannot increase the
  bubble radius of $w$. The result generalizes to~\eqref{eq:monotonicity} in a
  straightforward way.
  We now show the sub-modularity of the objective function, i.e.,
  that~\eqref{eq:submod} holds. We start by showing that, for $w \in
  \pol{C}{G}$, it holds
  \begin{align*}
    \bubblechange{G}{w}{\{e_v, e_u\}}{\{m_v, m_u\}}{t} \le &
    \bubblechange{G}{w}{e_v}{m_v}{t} \\
    &+ \bubblechange{G}{w}{e_u}{m_u}{t} \enspace.
  \end{align*}
  With an expansion of the definition and a slight rearrangement of the terms,
  the above inequality is equivalent to
  \[
    \underbracket{\bubble{G_v}{w}{t} -
    \bubble{G_{vu}}{w}{t}}_{\bubblechange{G_v}{w}{e_u}{m_u}{t}} \le
    \underbracket{\bubble{G}{w}{t} -
    \bubble{G_{u}}{w}{t}}_{\bubblechange{G}{w}{e_u}{m_u}{t}},
  \]
  i.e., the gain of adding the same edge (in this case $e_u$) is smaller when
  the edge is added to a graph (in this case $G_v$) that has a superset of the
  edges (compared to $G$).

  Consider all the walks from $w$ that either pass through $v$ or $u$. Among
  such walks, let $\mathcal{E}_v$ be the event of seeing $v$ first and
  $\mathcal{E}_u$ be the event of seeing $u$ first. If a walk does not pass
  through either $v$ or $u$, its probability of hitting the other color is the
  same in all graphs, as the graphs differ only in the outgoing edges from these
  two nodes and their weights. For the same reason, $\cP(\mathcal{E}_v)$ and
  $\cP(\mathcal{E}_u)$ do not change across the graphs. Thus,
  \begin{align*}
    &\bubble{G_v}{w}{t} - \bubble{G_{vu}}{w}{t} = \\
    &\left( \expect{G_v}{\ttime{t}{w}{\bar{C}_w} \mid \mathcal{E}_v} -
    \expect{G_{vu}}{\ttime{t}{w}{\bar{C}_w} \mid \mathcal{E}_v} \right)
    \cP(\mathcal{E}_v)\\
    &+ \left( \expect{G_v}{\ttime{t}{w}{\bar{C}_w} \mid \mathcal{E}_u} -
    \expect{G_{vu}}{\ttime{t}{w}{\bar{C}_w} \mid \mathcal{E}_u} \right)
    \cP(\mathcal{E}_u) \enspace.
  \end{align*}
  Similarly,
  \begin{align*}
    &\bubble{G}{w}{t} - \bubble{G_u}{w}{t} = \\
    &\left( \expect{G}{\ttime{t}{w}{\bar{C}_w} \mid \mathcal{E}_v} -
    \expect{G_u}{\ttime{t}{w}{\bar{C}_w} \mid \mathcal{E}_v} \right)
    \cP(\mathcal{E}_v)\\
    &+ \left( \expect{G}{\ttime{t}{w}{\bar{C}_w} \mid \mathcal{E}_u} -
    \expect{G_u}{\ttime{t}{w}{\bar{C}_w} \mid \mathcal{E}_u} \right)
    \cP(\mathcal{E}_u) \enspace.
  \end{align*}
  We want to show that it holds
  \begin{align}
    \expect{G_v}{\ttime{t}{w}{\bar{C}_w \mid \mathcal{E}_v}} -
    \expect{G_{vu}}{\ttime{t}{w}{\bar{C}_w \mid \mathcal{E}_v}} \nonumber\\
    \le \expect{G}{\ttime{t}{w}{\bar{C}_w \mid \mathcal{E}_v}} -
    \expect{G_u}{\ttime{t}{w}{\bar{C}_w \mid \mathcal{E}_v}} \enspace.\label{eq:diffcondonv}
  \end{align}
  and
  \begin{align}
    \expect{G_v}{\ttime{t}{w}{\bar{C}_w \mid \mathcal{E}_u}} -
    \expect{G_{vu}}{\ttime{t}{w}{\bar{C}_w \mid \mathcal{E}_u}} \nonumber\\
    \le \expect{G}{\ttime{t}{w}{\bar{C}_w \mid \mathcal{E}_u}} -
    \expect{G_u}{\ttime{t}{w}{\bar{C}_w \mid \mathcal{E}_u}}
    \enspace.\label{eq:diffcondonu}
  \end{align}
  We can write
  \begin{align*}
    &\expect{G}{\ttime{t}{w}{\bar{C}_w} \mid \mathcal{E}_v} \\
    &= \sum_{i=1}^t \left (
    i + \expect{G}{\ttime{t-i}{v}{\bar{C}_v}} \right ) \cP\left(
    \towithin{w}{v}{=i}{G} \mid \mathcal{E}_v \right) \enspace.
  \end{align*}
  The probability on the right is the same on all graphs. Similar expressions
  hold for
  \[
    \expect{G_v}{\ttime{t}{w}{\bar{C}_w} \mid \mathcal{E}_v},\  
    \expect{G_u}{\ttime{t}{w}{\bar{C}_w} \mid \mathcal{E}_v}, \  
    \expect{G_{vu}}{\ttime{t}{w}{\bar{C}_w} \mid \mathcal{E}_v},
  \]
  and when conditioning on $\mathcal{E}_u$. 
  To prove~\eqref{eq:diffcondonv} and~\eqref{eq:diffcondonu}, we now show
  that, for every $t' \le t$, it holds
  \begin{equation}\label{eq:techv}
    \expect{G_v}{\ttime{t'}{v}{\bar{C}_v}} -
    \expect{G_{vu}}{\ttime{t'}{v}{\bar{C}_v}} \leq
    \expect{G}{\ttime{t'}{v}{\bar{C}_v}} -
    \expect{G_u}{\ttime{t'}{v}{\bar{C}_v}},
  \end{equation}
  and
  \begin{equation*}
    \expect{G_v}{\ttime{t'}{u}{\bar{C}_u}} -
    \expect{G_{vu}}{\ttime{t'}{u}{\bar{C}_u}} \!\leq\! 
    \expect{G}{\ttime{t'}{u}{\bar{C}_u}} -
    \expect{G_u}{\ttime{t'}{u}{\bar{C}_u}}.
  \end{equation*}
  We focus on showing~\eqref{eq:techv}, as the same steps, with simple
  modifications, can be followed to show the other inquality. For $Z \in \{G,
  G_u, G_v, G_{vu}\}$, let $\mathcal{A}_Z$ be the event that a random walk
  starting at $v$ reaches $u$ in at most $t$ steps before visiting any vertex in
  $\bar{C}_v$, and let $\bar{\mathcal{A}}_Z$ be the complementary event. It holds
  \[
    \cP(\mathcal{A}_{G_v}) = \cP(\mathcal{A}_{G_{vu}}) \le \cP(\mathcal{A}_{G})
    = \cP(\mathcal{A}_{G_{u}}),
  \]
  due to the insertion of $e_v$. It also holds
  \[
    \expect{G_v}{\ttime{t'}{v}{\bar{C}_v} \mid \bar{\mathcal{A}}_{G_v}} =
    \expect{G_{vu}}{\ttime{t'}{v}{\bar{C}_v} \mid \bar{\mathcal{A}}_{G_{vu}}},
  \]
  and
  \[
    \expect{G}{\ttime{t'}{v}{\bar{C}_v} \mid \bar{\mathcal{A}}_{G}} =
    \expect{G_{u}}{\ttime{t'}{v}{\bar{C}_v} \mid \bar{\mathcal{A}}_{G_{u}}},
  \]
  Using the law of total expectation (across $\mathcal{A}_Z$ and
  $\bar{\mathcal{A}}_Z$) and applying these facts, we can
  rewrite~\eqref{eq:techv} as
  \begin{align*}
    &\left( \expect{G_v}{\ttime{t'}{v}{\bar{C}_v} \mid \mathcal{A}_{G_v}} -
    \expect{G_{vu}}{\ttime{t'}{v}{\bar{C}_v} \mid \mathcal{A}_{G_{vu}}} \right)
    \cP(\mathcal{A}_{G_v})\\
    \le &\left( \expect{G}{\ttime{t'}{v}{\bar{C}_v} \mid \mathcal{A}_{G}} -
    \expect{G_{u}}{\ttime{t'}{v}{\bar{C}_v} \mid \mathcal{A}_{G_{u}}} \right)
    \cP(\mathcal{A}_{G}) \enspace.
  \end{align*}
  The differences between parentheses have the same value, as their
  corresponding terms have the same values. The inequality holds because
  $\cP(\mathcal{A}_{G_v}) \le \cP(\mathcal{A}_{G})$ due to the insertion of
  $e_v$ in $G$ to obtain $G_v$.


\end{proof}

\begin{figure*}[t]
\begin{center}
\begin{subfigure}[b]{0.22\textwidth}
    \includegraphics[width=\textwidth]{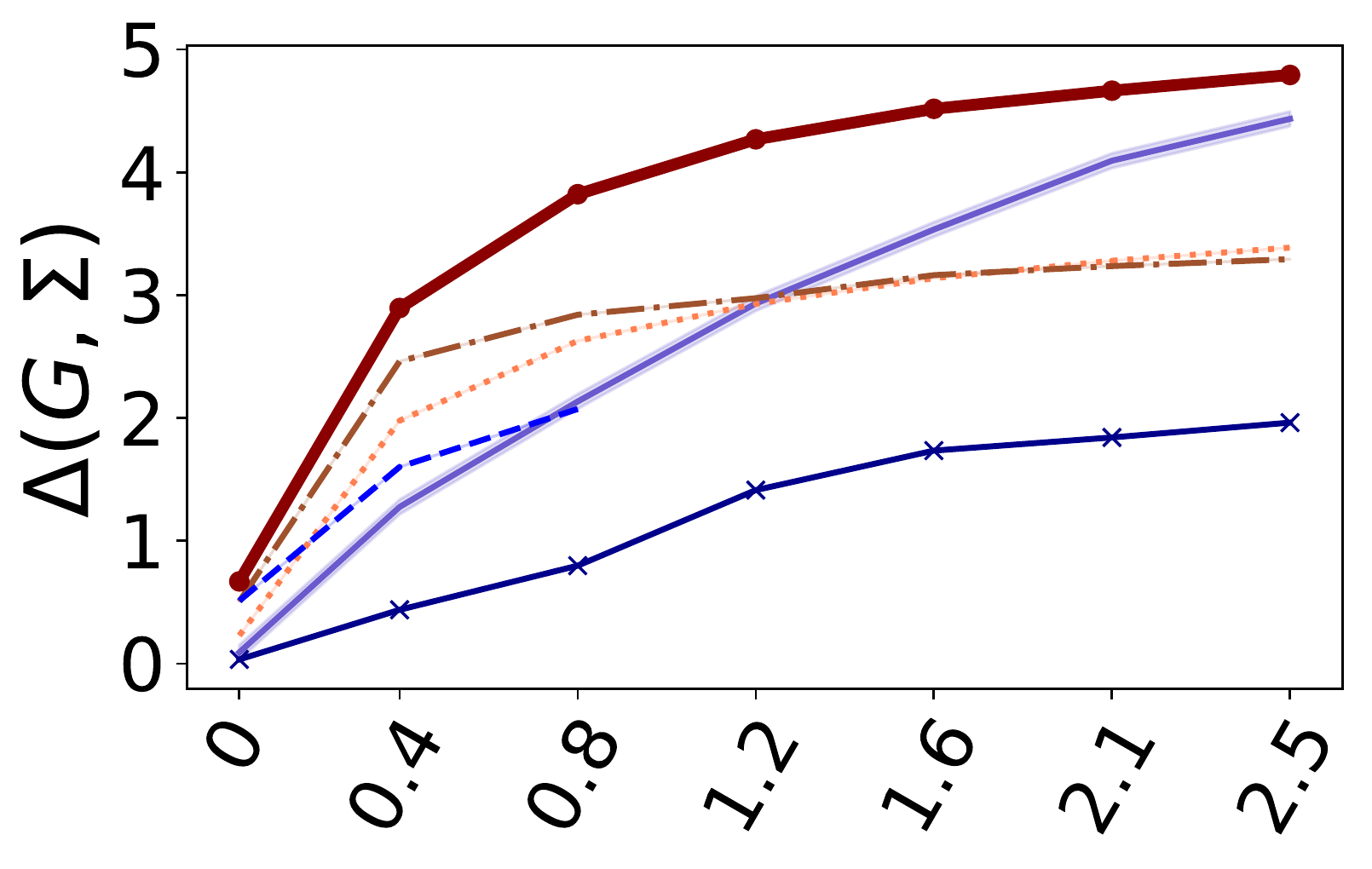}
    \includegraphics[width=\textwidth]{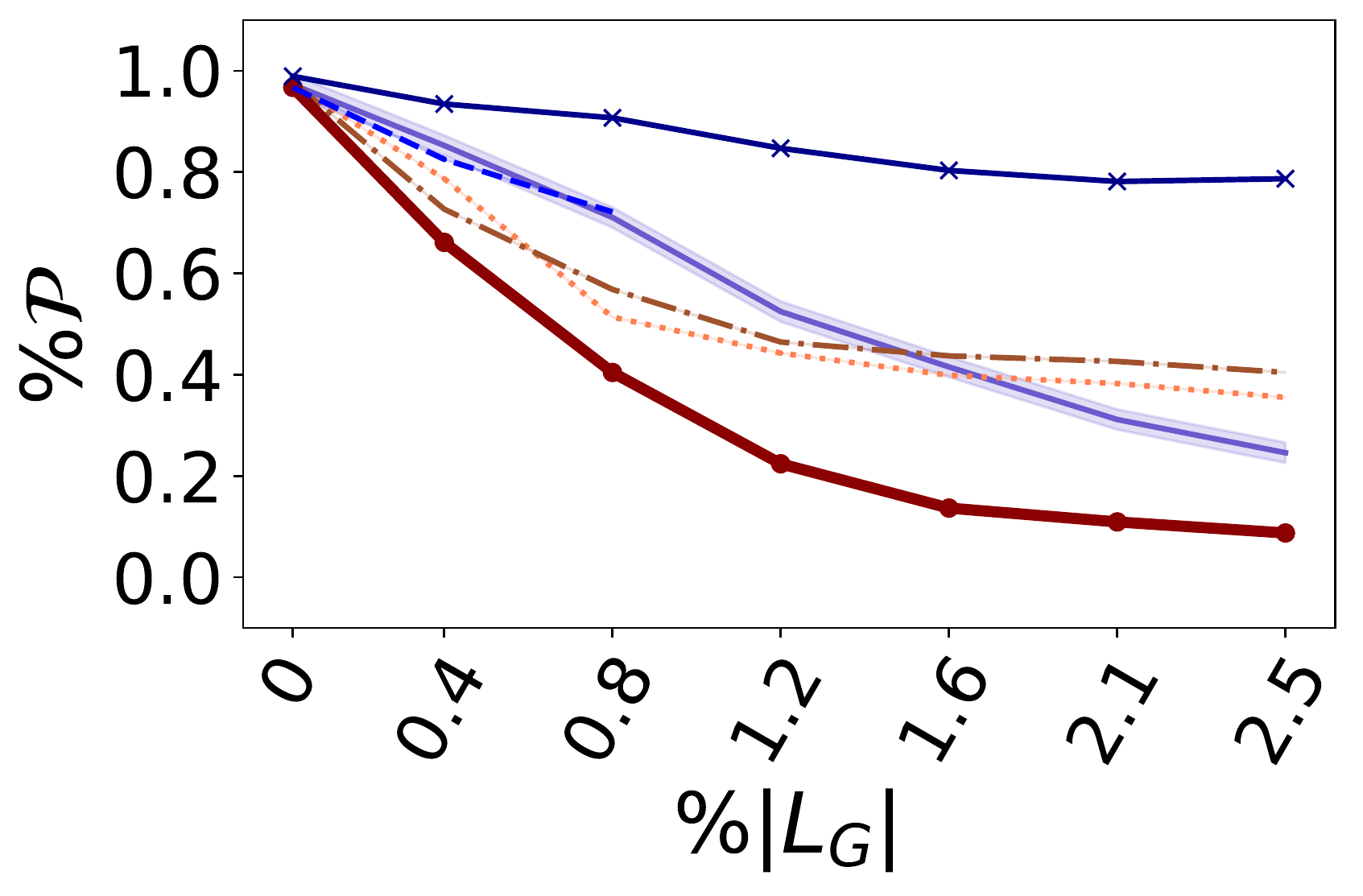}
    \ifextended%
    \else
    \fi
    \caption{Guns}
\end{subfigure}%
\begin{subfigure}[b]{0.22\textwidth}
    \includegraphics[width=\textwidth]{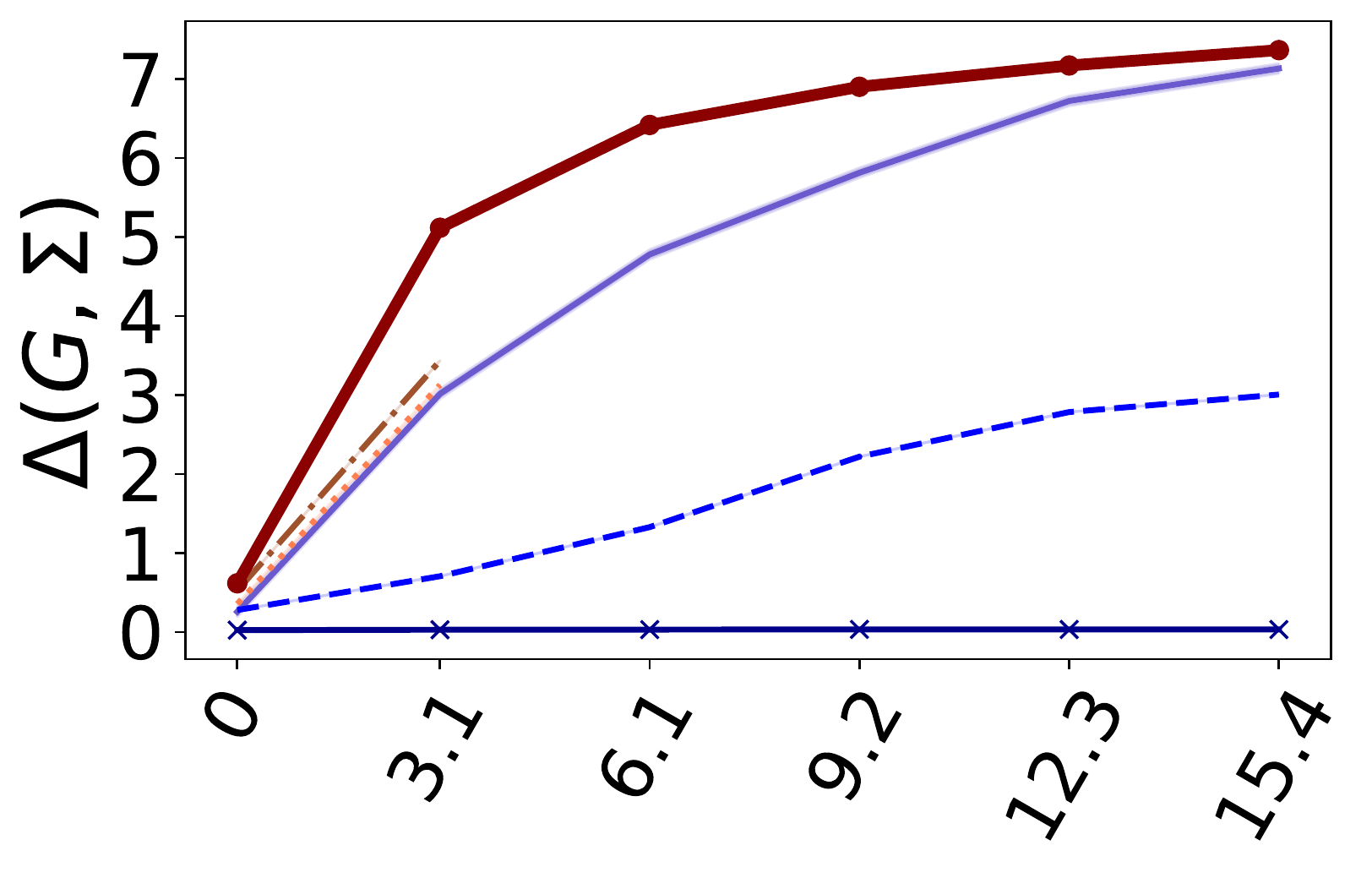}
    \includegraphics[width=\textwidth]{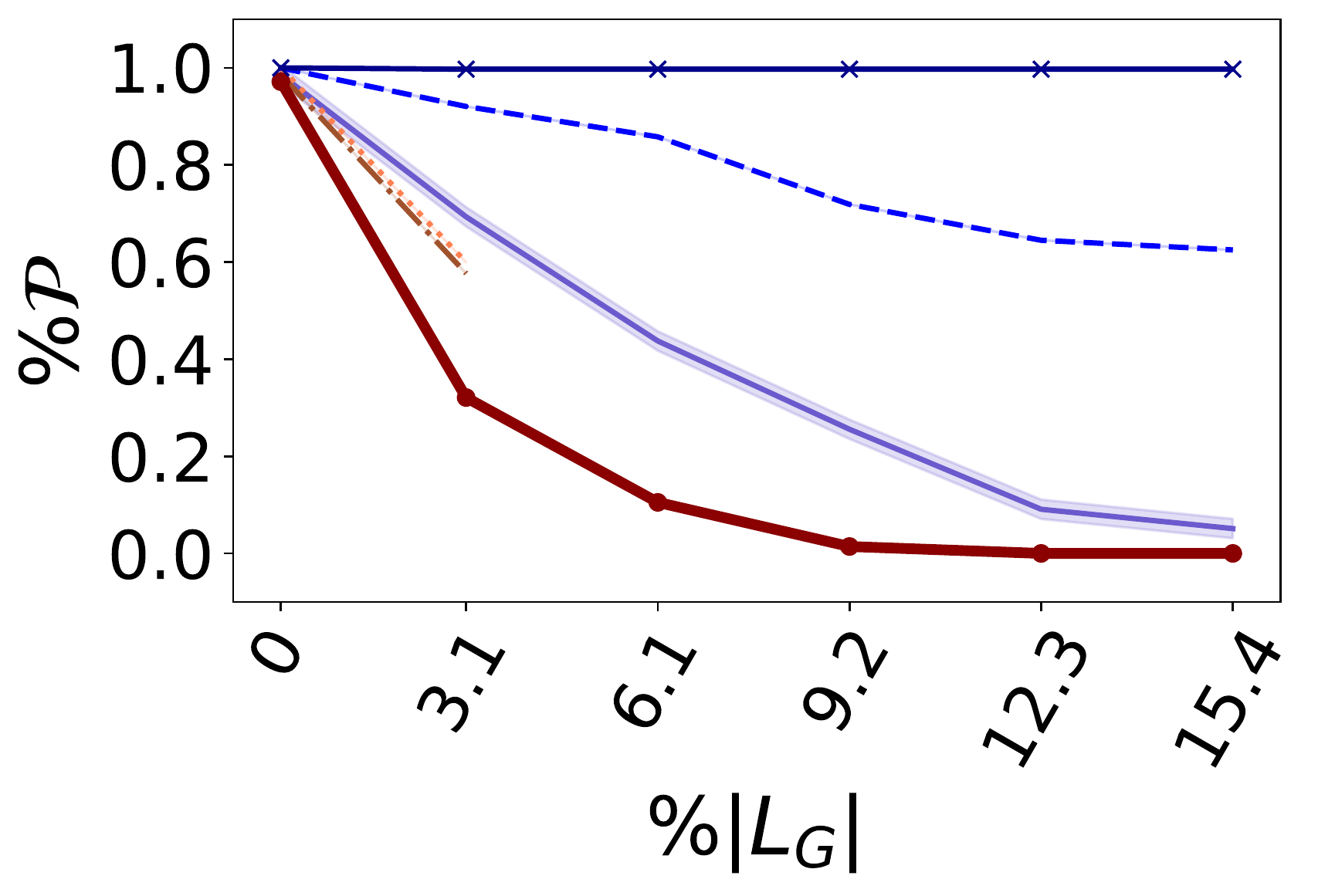}
    \ifextended%
    \else
    \fi
    \caption{MaAs}
\end{subfigure}

\begin{subfigure}[b]{0.22\textwidth}
    \includegraphics[width=\textwidth]{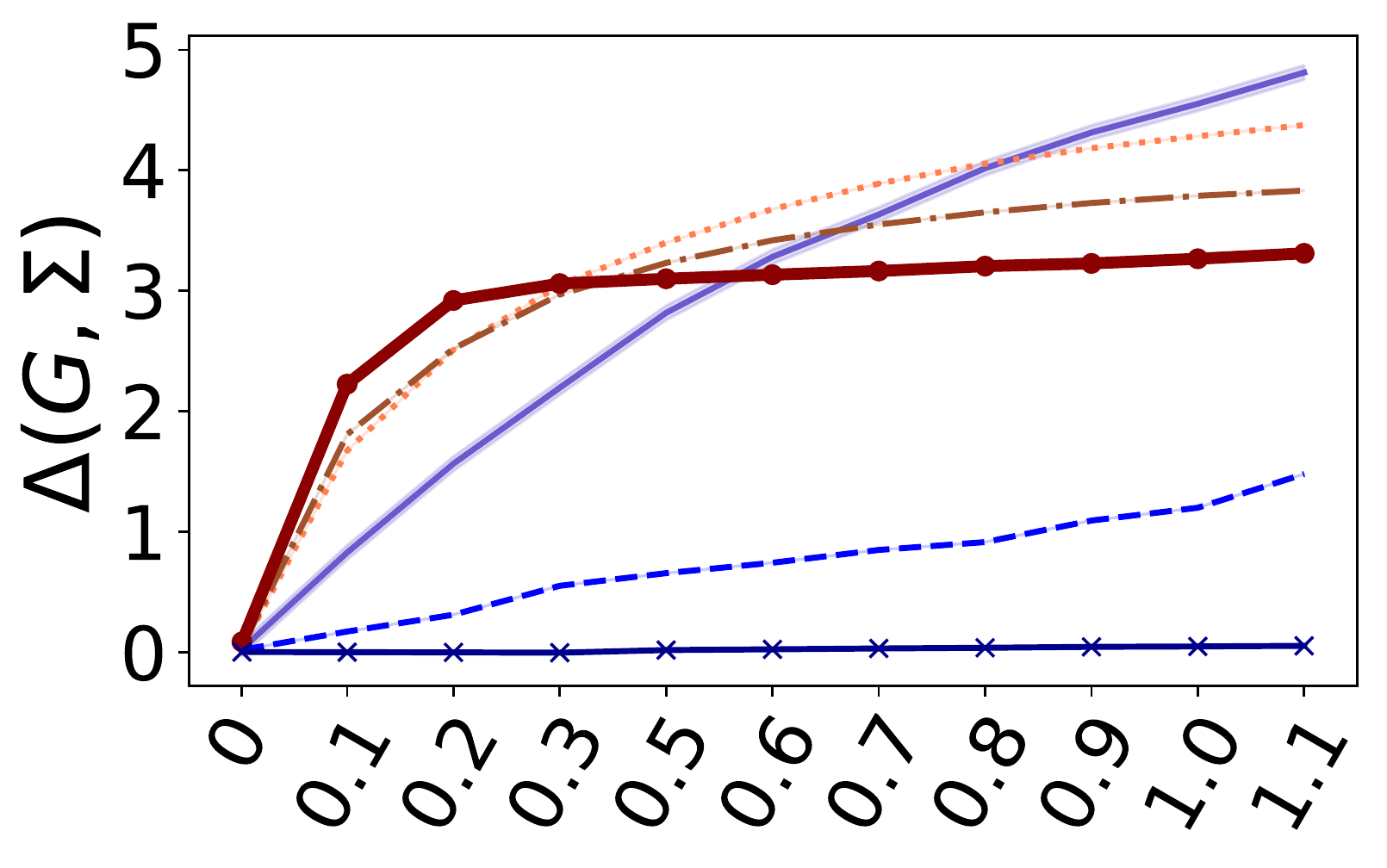}
    \includegraphics[width=\textwidth]{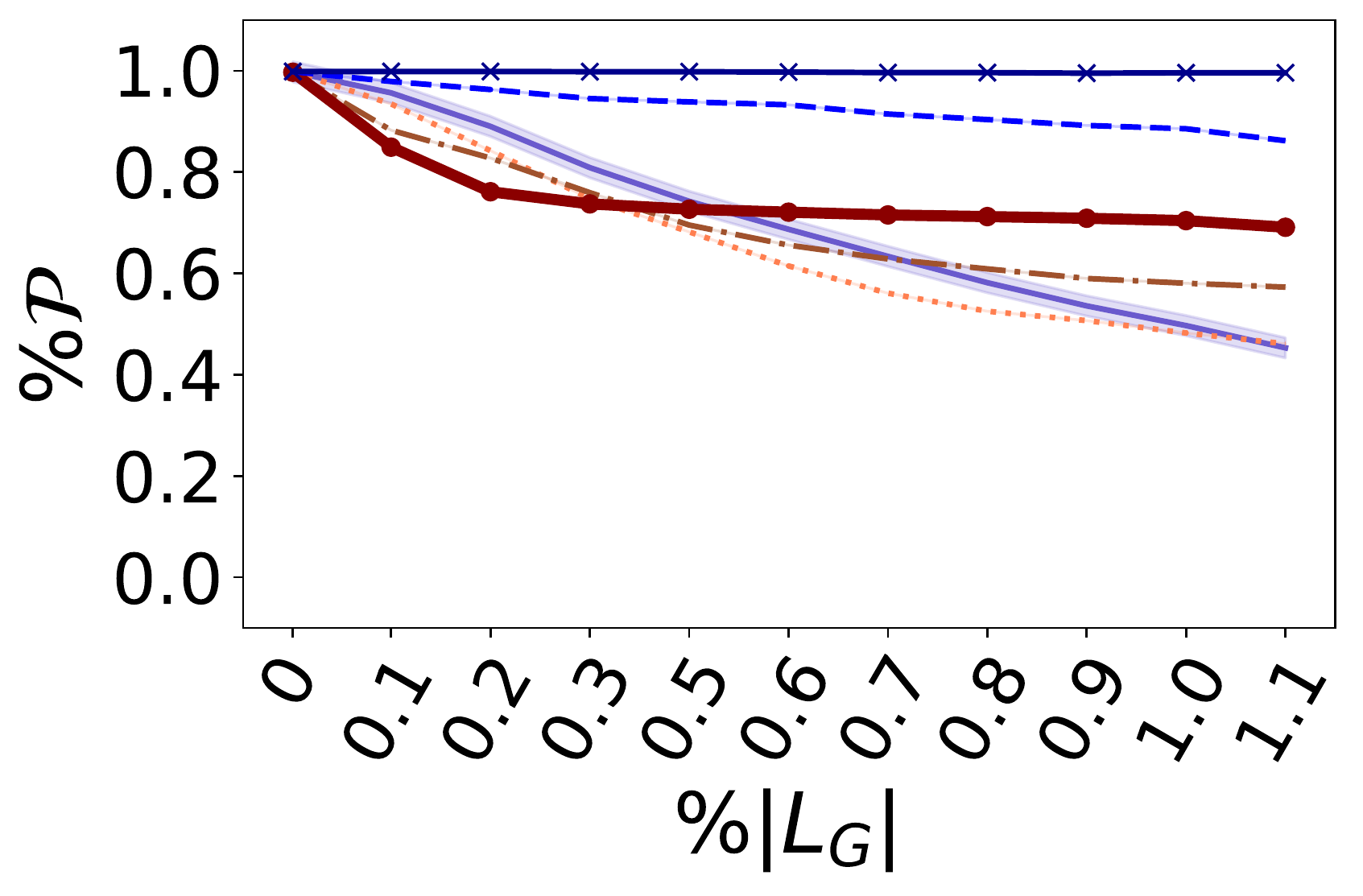}
    \ifextended%
    \else
    \fi
    \caption{Sociology}
\end{subfigure}
\begin{subfigure}[b]{0.22\textwidth}
    \includegraphics[width=\textwidth]{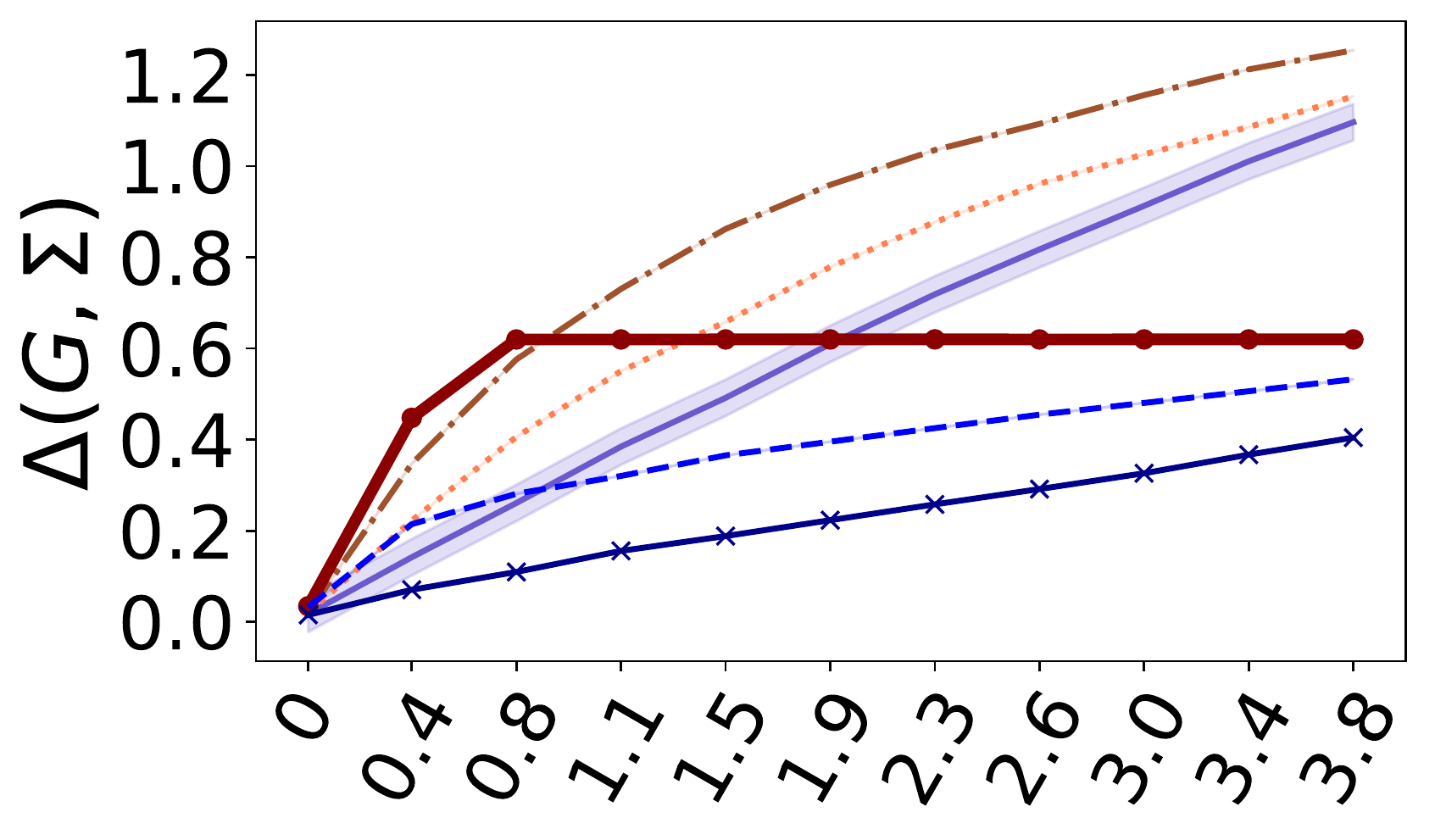}
    \includegraphics[width=\textwidth]{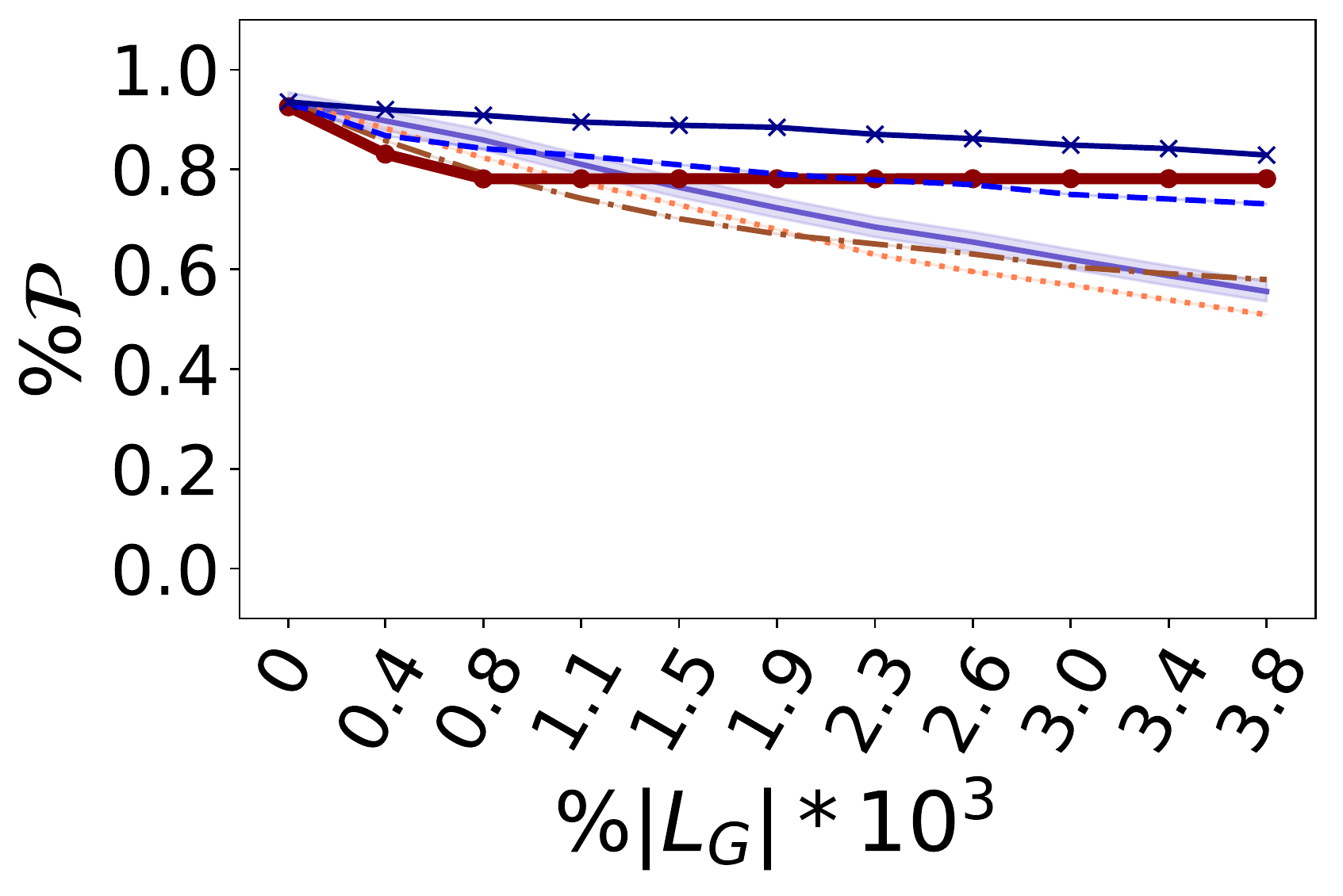}
    \ifextended%
    \else
    \fi
    \caption{Politics}
\end{subfigure}
\begin{subfigure}[b]{0.63\textwidth}
    \includegraphics[width=\textwidth]{legend.pdf}
\end{subfigure}
\end{center}
\ifextended%
\else
\fi
\caption{The first row shows the $\Delta(G, \Sigma)$ (y-axis) for increasing value of
$k$, reported in terms of $\%\mathcal{L}_G$, the union of possible edges across  $\badc{C}{G}$ and $\bar{C}$ for $C \in {R,B}$, (x-axis) for each algorithm. Higher values of
$\Delta$ show more significant reduction of the structural bias. In the second
row, we show the percentage of nodes that are still parochial, $\%\mathcal{P} =
\frac{\card{\pol{}{G}}-\card{\pol{}{\Gnew}}}{\card{\pol{}{G}}}$
after $k$ additions.}
\Description{Results of the experiments for different graphs. See the caption
and the text for description of these results.}
\end{figure*}